\documentclass{IEEEtran}

\usepackage[utf8]{inputenc} 
\usepackage[T1]{fontenc}
\usepackage[cmex10]{amsmath}
\usepackage{graphicx, graphics, dsfont}
\usepackage{amssymb, amsfonts, amsthm}
\usepackage{subcaption}

\newtheorem{theorem}{Theorem}
\newtheorem{corollary}{Corollary}

\newtheorem{proposition}{Proposition}

\newcommand{\qa}{{\bf a}}

\newcommand{\qg}{{\bf g}}

\newcommand{\qr}{{\bf r}}

\newcommand{\qw}{{\bf w}}
\newcommand{\qx}{{\bf x}}
\newcommand{\qy}{{\bf y}}
\newcommand{\qz}{{\bf z}}

\newcommand{\qG}{{\bf G}}
\newcommand{\qH}{{\bf H}}
\newcommand{\qI}{{\bf I}}

\newcommand{\qR}{{\bf R}}
\newcommand{\qS}{{\bf S}}

\newcommand{\PP}{\mathbb{P}}
\newcommand{\E}{\mathbb{E}}

\newcommand{\ee}{\bar{\eta}}

\newcommand{\eh}{\hat{\eta}}
\newcommand{\h}{\hat{h}}

\allowdisplaybreaks
\IEEEoverridecommandlockouts

\begin{document}
\title{On the Diversity and Coded Modulation Design\\ of Fluid Antenna Systems}
\author{Constantinos Psomas, \IEEEmembership{Senior Member, IEEE}, Ghassan M. Kraidy, \IEEEmembership{Senior Member, IEEE}, Kai-Kit Wong, \IEEEmembership{Fellow, IEEE}, and Ioannis Krikidis, \IEEEmembership{Fellow, IEEE}\vspace{-1mm}
\thanks{The work of C. Psomas, G. M. Kraidy, and I. Krikidis has received funding from the European Research Council (ERC) under the European Union's Horizon 2020 research and innovation programme (Grant agreement No. 819819). It was also co-funded by the European Regional Development Fund and the Republic of Cyprus through the Research and Innovation Foundation, under the project INFRASTRUCTURES/1216/0017 (IRIDA). The work of K. K. Wong was supported by the Engineering and Physical Sciences Research Council (EPSRC) under grant EP/W026813/1.}
\thanks{C. Psomas and I. Krikidis are with the Department of Electrical and Computer Engineering, University of Cyprus, Nicosia, Cyprus (e-mail: \{psomas, krikidis\}@ucy.ac.cy). G. M. Kraidy is with the Department of Electronic Systems, Norwegian University of Science and Technology, Norway (e-mail: ghassan.kraidy@ntnu.no). K.-K. Wong is with the Department of Electronic and Electrical Engineering, University College London, London, United Kingdom (e-mail: kai-kit.wong@ucl.ac.uk).}
\thanks{Part of this work was presented at the IEEE International Conference on Communications, Rome, Italy, May 2023 \cite{icc}.}}

\maketitle

\begin{abstract}
Reconfigurability is a desired characteristic of future communication networks. From a transceiver's standpoint, this can be materialized through the implementation of fluid antennas (FAs). An FA consists of a dielectric holder, in which a radiating liquid moves between pre-defined locations (called ports) that serve as the transceiver's antennas. Due to the nature of liquids, FAs can practically take any size and shape, making them both flexible and reconfigurable. In this paper, we deal with the outage probability of FAs under general fading channels, where a port is scheduled based on selection combining. An analytical framework is provided for the performance with and without errors due to post-scheduling delays. We show that although FAs achieve maximum diversity, this cannot be realized in the presence of delays. Hence, a linear prediction scheme is proposed that overcomes delays and restores the lost diversity by predicting the next scheduled port. Moreover, we design space-time coded modulations that exploit the FA’s sequential operation with space-time rotations and code diversity. The derived expressions for the pairwise error probability and average word error rate give an accurate estimate of the performance. We illustrate that the proposed design attains maximum diversity, while keeping a low-complexity receiver, thereby confirming the feasibility of FAs.
\end{abstract}


\section{Introduction}
The numerous advances in antenna technology and in particular antenna arrays have been important in the evolution of communication systems towards 5G and beyond 5G networks \cite{ZHANG}. Indeed, the implementation of multiple-input multiple-output (MIMO) antenna architectures has been an essential element in wireless networks for the realization of high data rates and spectral efficiency due to beamforming and spatial multiplexing. Such antennas are usually made of metal and are designed in such a way so as to meet specific network requirements. Furthermore, their design is subject to physical constraints, the most significant being the spacing between two antennas, which needs to be at least as half as the carrier's wavelength to avoid electromagnetic coupling \cite{HEATH}. Naturally, this metallic structure makes them static (i.e., inflexible), impractical and too costly for very small devices to have many antennas.

Recently, there has been several efforts to introduce reconfigurability in wireless networks. Towards this end, reconfigurable intelligent surfaces were proposed to control the propagation environment via software-controlled metasurfaces \cite{LIASKOS}. From a transceiver's point-of-view, the notion of fluid antennas (FAs), also known as liquid antennas, has been recently proposed in order to add both flexibility and reconfigurability at the radio frequency (RF) front-end \cite{HUANG}. In particular, FAs consist of radiating liquid elements such as Mercury, eutectic gallium-indium (EGaIn) and even sea water, enclosed in a dielectric holder \cite{PARACHA}. The holder contains several pre-defined positions, known as ports, where the employed liquid can be moved towards a selected port in a programmable and controllable manner. Therefore, this technology provides new degrees of freedom in the design of wireless communication systems. In particular, an FA uses just one RF chain and thus the spacing constraint does not apply in this case \cite{KIT}. As such, an FA can be small and energy efficient, making it a suitable technology for wireless devices and sensors \cite{KIT3}. In fact, given their flexibility, they are exceptionally ideal for wearable devices. Moreover, due to the physical characteristics of liquids, FAs can modify their shape but also adapt their position in order to reconfigure the operating frequency \cite{AS, MK}, the radiation pattern \cite{KIT6} as well as the gain and the polarization \cite{KIT7}, which can be beneficial in multi-user networks for interference mitigation.

Despite the fact that FAs have been studied from an RF/microwave engineering perspective and several early prototypes exist in the literature, e.g. see \cite{HUANG, PARACHA, CHEN, XING}, the theoretical foundations of FAs and the investigation of communication techniques that unlock their full potentials are still not understood. Indeed, the exploitation of the liquid dimension associated with the FAs will open new design opportunities and establish a new communication paradigm \cite{KIT, KIT2, KIT4}. In \cite{KIT}, the authors consider an FA within a linear space, where the selected port is based on the selection combining scheme. They study the performance of an FA system in terms of outage probability and show that an FA can outperform a maximum ratio combining (MRC) system with conventional antennas, when the number of ports is sufficiently large. The work in \cite{KIT2} extends the study with respect to the ergodic capacity, where it is demonstrated that FAs can match the capacity of MRC systems. Finally, the performance of FAs with multiuser interference is studied in \cite{KIT4}; it is shown that with a large enough number of ports, the FA attains a relatively low outage probability. Now, the aforementioned studies assume that the performance is not affected by any imperfections as a result of the selection process. Recent advances in microfluidics have certainly made the realization of reconfigurable RF devices with low transition times between ports possible \cite{NC}. Moreover, the concept of a reconfigurable FA can be implemented potentially through the employment of switchable non-fluid pixels \cite{DR}, where the transition can be accomplished almost instantly. However, due to the sequential nature of FAs, estimation and selection will be affected by some delays between the pre-scheduling and the post-scheduling of a port, especially for a large number of ports. Therefore, a proper analysis considering processing imperfections is of great importance.

Motivated by this, in this paper, we study the outage probability of an FA system, where the channel at the selected port is subject to practical delays. By taking into account spatial correlation, we analytically determine the loss in diversity and propose a prediction scheme to restore the performance. To the best of the authors' knowledge, the provided analytical framework and methodology under spatial correlation is novel. Now, selection diversity techniques have been widely used due to the simplicity in their implementation and in the design of transmission schemes, as all the signals are placed on one channel. On the other hand, combining diversity techniques provide better performance as they exploit all the available channels (ports) for transmission, which however implies a more sophisticated transmission scheme design and multiple RF chains to activate multiple ports. Therefore, we also deal with the design of a low-complexity coded modulation scheme for FA systems. In particular, we propose a novel port combining scheme, which exploits the liquid's movement through the holder and achieves maximum code diversity. In uncoded transmission over block-fading channels such as multiple-antenna \cite{belfiore05} or cooperative channels \cite{yang07}, achieving optimal performance requires space-time coding with full spreading (i.e. same space and time dimensions), which increases the cardinality of the transmit symbol vector and thus results in high demodulation complexity. However, when channel coding is implemented with optimal interleaving, space-time rotations with limited spreading (i.e. time dimension smaller than space dimension) can be used to complement the code diversity, thus leading to a low-complexity receiver \cite{gresset08}. Therefore, code diversity can mitigate the drawback of non-ergodic fading channels, in that a proper interleaving of bits over channel states allows to both recover diversity and enhance the coding gain without additional complexity. Specifically, the contributions of this paper are threefold.
\begin{itemize}
\item We study the performance of FA systems, where a port is scheduled based on selection combining, under general Nakagami fading channels. Three basic FA architectures are taken into account, namely, linear, circular and wheel-shaped, which exhibit different correlation patterns between the ports. We derive analytical expressions for the outage probability and show that the architecture with the best performance differs, depending on the FA size. Moreover, the provided asymptotic expressions quantify the system's diversity and outage gain. An FA achieves full channel (spatial) diversity, which is independent of the architecture, whereas its outage gain depends on the correlation pattern.
\item By considering errors due to delays, we present an analytical framework for deriving the outage probability of outdated channel estimates. As a result of the post-scheduling errors, the performance deteriorates significantly both in diversity as well as in outage gain; indeed, in some cases, the performance is inversely proportional to the number of ports. To recover the performance loss, we provide a linear prediction scheme based on channel knowledge obtained from previous training resource blocks. We show that with just a small number of training blocks, the performance is restored and maximum diversity can be obtained.
\item A port combining scheme is proposed, where the ports are sequentially activated for reception. Specifically, we combine space-time rotations with code diversity to design space-time coded modulations that have optimal performance over FA correlated block-fading channels. Unlike previous designs, we consider scenarios in which not all modulated symbols in the transmission vector are space-time rotated by introducing an average spreading factor. This will allow to achieve maximum code diversity while further reducing the receiver's complexity. Expressions for the pairwise error probability (PEP) and average word error rate (WER) are derived that give an accurate estimate for the performance of the coded modulations.
\end{itemize}

The rest of the paper is organized as follows. Section \ref{sys_model} describes the considered system model. Section \ref{sec_outdated} presents the analytical framework for the performance of port selection with pre-scheduling, post-scheduling and predicted channels. Section \ref{sec_coded} provides the coded modulation design and WER performance for port combining. Finally, Section \ref{numerical} presents the numerical results and Section \ref{conclusion} concludes the paper.

{\it Notation}: Lower and upper case boldface letters denote vectors and matrices, respectively; $[\cdot]^\top$ and $[\cdot]^\dag$ denote the transpose and the transpose conjugate, respectively; $||\qH||$ gives the Frobenius norm of matrix $\qH$; $\qI_N$ denotes the $N\times N$ identity matrix; $\text{diag}(a_1, \dots, a_N)$ denotes a diagonal matrix with entries $a_1,\dots,a_N$ on the main diagonal; $\PP\{X\}$ and $\E\{X\}$ represent the probability and expectation of $X$, respectively; $Q(\cdot)$ is the $Q$-function; $\Gamma(\cdot)$, $\Gamma(\cdot,\cdot)$ and $\gamma(\cdot,\cdot)$ denote the complete, upper, and lower incomplete gamma function, respectively; $Q_m(\cdot,\cdot)$ denotes the Marcum-$Q$ function of order $m$; $J_n(\cdot)$ and $I_n(\cdot)$ are the Bessel function and the modified Bessel function, respectively, of the first kind and order $n$ \cite{GRAD}; $\jmath = \sqrt{-1}$ is the imaginary unit; $\mathds{1}_X$ is the indicator function of $X$ with $\mathds{1}_X = 1$, if $X$ is true, and $\mathds{1}_X = 0$, otherwise.

\section{System Model}\label{sys_model}
Consider a point-to-point network, with a conventional single-antenna transmitter and a single-FA receiver. The FA uses a single RF chain and consists of $N$ ports, evenly distributed over a dielectric holder of a specific topological space, defined in the next sub-section. We assume that the FA can switch on a single port by displacing the employed liquid to its location with a mechanical pump \cite{HUANG}. The transmitter utilizes a fixed transmission power $P$.

\begin{figure*}[t]\centering	\subfloat[Linear.]{\includegraphics[height=4.4cm]{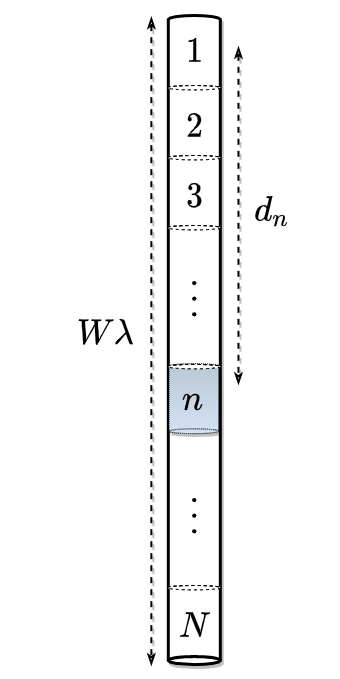}\label{fig1a}}\hspace{12mm}
\subfloat[Circular.]{\includegraphics[height=4.4cm]{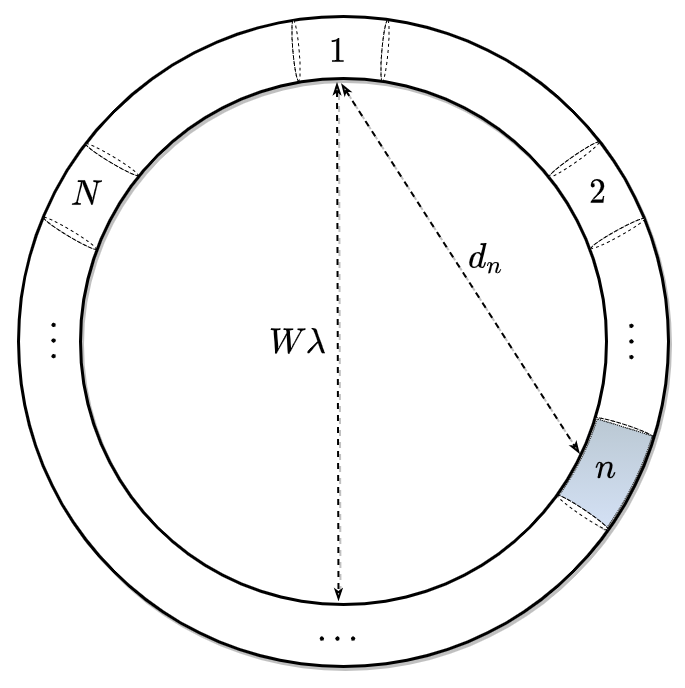}\label{fig1b}}\hspace{12mm}
\subfloat[Wheel.]{\includegraphics[trim=40 60 50 60, clip,height=4.4cm]{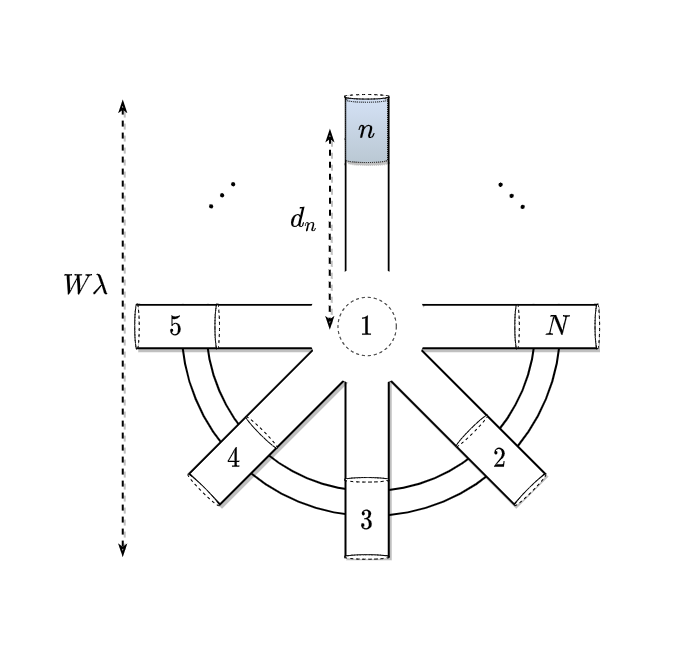}\label{fig1c}}
\caption{Considered FA architectures.}\label{fig1}
\end{figure*}

\subsection{FA Architectures}\label{fa_arch}
We consider three FA architectures\footnote{These three topologies are considered for the sake of studying the effect of different correlation patterns on the performance of FAs. Nevertheless, the presented analytical framework is general and is not limited to these specific geometric topologies.}: a uniform linear array \cite{HEATH, KIT} (Fig. \ref{fig1}\subref{fig1a}), a uniform circular (toroidal) array \cite{HEATH, XING} (Fig. \ref{fig1}\subref{fig1b}), and a wheel-shaped array \cite{LIN} (Fig. \ref{fig1}\subref{fig1c}). The size of each topology is characterized by the value $W\lambda$, where $\lambda$ is the wavelength of the transmitted signals. Moreover, due to their specific topological attributes, each architecture exhibits different displacements values between the ports. Specifically, the displacement (i.e. the Euclidean distance) at the $n$-th port from the first one is given by
\begin{align}
d_n = \frac{n-1}{N-1} W\lambda,
\end{align}
for the linear topology \cite{KIT}. In this case, the displacement increases with $n$ and so the maximum displacement is unique and equal to $d_N$. Moreover, all $N$ displacements are distinct, i.e. $d_1 \neq d_2 \neq \cdots \neq d_N$. For the circular topology, the displacement can be written as
\begin{align}
d_n = \sin\left(\frac{n-1}{N}\pi\right)W\lambda,
\end{align}
which essentially corresponds to the circle's chord formed by the two ports. Here, for $N$ even, the maximum displacement is unique and equal to $d_{\frac{N}{2}+1}$, whereas for odd $N$, the maximum displacements are $d_{\frac{N+1}{2}} = d_{\frac{N+3}{2}}$. Also, the remaining displacements are pairwise equal, that is, $d_n = d_{N-n+2}$. Finally, for the wheel-shaped topology, we have
\begin{align}
d_n = \mathds{1}_{n>1}\frac{W\lambda}{2},
\end{align}
where it is clear that it is independent of the port's location for $n > 1$, i.e. $d_2 = d_3 = \cdots = d_N$. The displacements are illustrated in Fig. \ref{fig1} for each topology.

\subsection{Channel Model}
All wireless links are assumed to exhibit Nakagami fading with integer parameter $m \geq 1$; $m=1$ corresponds to Rayleigh fading and $m \to \infty$ to deterministic no-fading scenario. We represent the Nakagami-$m$ fading channels at the $N$ ports by a set of $Nm$ complex Gaussian random variables, as follows \cite{NB}
\begin{align}
g_{1,k} &= x_{1,k} + \jmath y_{1,k},\label{rv1}\\
g_{n,k} &= \rho_n g_{1,k} + \sqrt{1-\rho_n^2} (x_{n,k} + \jmath y_{n,k}),\label{rv2}
\end{align}
where $x_{n,k}$ and $y_{n,k}$ are independent Gaussian random variables with zero mean and variance $\sigma^2/2$, $\sigma > 0$, i.e. $x_{n,k}, y_{n,k} \sim \mathcal{N}(0,\sigma^2/2)$, for $k = 1,\dots, m$ and $n = 1,\dots,N$. Let $h_n$, $n = 1,\dots,N$, be the normalized sum of the squared magnitude of $g_{n,k}$, that is,
\begin{align}\label{sum}
h_n = \frac{1}{m} \sum_{k=1}^m |g_{n,k}|^2.
\end{align}
In this way, the random variables $\sqrt{h_n}$ are correlated Nakagami-$m$ fading envelopes with $\E\{h_n\} = \sigma^2$ \cite{NB}. So $h_n$ is the normalized sum of $m$ independent squared Rayleigh envelopes or, equivalently, the sum of $2m$ independent squared zero-mean Gaussian random variables with variance $\sigma^2/2$ \cite{NB}; throughout this paper, we will consider $\sigma^2 = 1$ for simplicity but the generalization to any $\sigma^2$ is straightforward and could be used to consider path-loss and other attenuation effects. By having the first port as the reference point, the correlation coefficient between $h_n$ and $h_1$ can be modeled by\footnote{Note that alternative spatial correlation models were recently proposed in \cite{KIT5,ALOUINI} which differ slightly from the one used in this paper. It would be an interesting future work to extend the work of this paper using the new models.} \cite{HEATH}
\begin{align}\label{correlation}
\rho_n = J_0\left(2\pi\frac{d_n}{\lambda}\right),
\end{align}
where $d_n$ is defined based on the considered topology. Note that the observations regarding the displacements for each topology also hold for the correlation coefficients. Finally, we assume that the channel coefficients are known to the FA receiver, but not the transmitter.

\section{Port Selection with Outdated Channel Estimates}\label{sec_outdated}
In this section, we study the performance of port selection in the considered FA system with outdated channel estimates. Specifically, the FA employs a selection combiner and thus chooses the port with the strongest received signal. Knowledge of the channels at each port is obtained through a training period, where each port is activated sequentially. Furthermore, we assume that the channel estimation is perfect. However, since the channel changes rapidly, the estimated (pre-scheduling) channels are subject to delays \cite{JLV}; these delays correspond to the duration needed for the liquid to be displaced at each port. As such, by the time the port with the best estimate is scheduled (i.e. switched on), that estimation may be outdated. We first focus on the performance of the pre-scheduling estimated channels and then of the post-scheduling (outdated) ones. Finally, we propose a prediction scheme that overcomes the negative effects imposed by the delays.

\subsection{Pre-scheduling Channel Estimates}
To facilitate the analysis, we let $\h_n$ denote the estimated received channel at the $n$-th port and $\h = \max\{\h_1,\h_2,\dots,\h_N\}$; even though the channel coefficients are a function of time, we drop the time-index for the sake of brevity. The theorem below provides the cumulative distribution function (CDF) of the estimated $\h$.

\begin{theorem}\label{thm1}
The CDF of the estimated $\h$ is given by
\begin{align}\label{cdf1}
F_{\h}(x) = \frac{1}{\Gamma (m)} \int_0^x \exp (-z) z^{m-1} \prod_{n=2}^N \phi_n(z,x) dz,
\end{align}
where
\begin{align}\label{phi}
\phi_n(z,x) = 1-Q_m\left(\sqrt{\frac{2 z \rho_n^2}{1-\rho_n^2}},\sqrt{\frac{2x}{1-\rho_n^2}}\right),
\end{align}
and $\rho_n$ is given by \eqref{correlation}.
\end{theorem}

\begin{proof}
See Appendix \ref{thm1_prf}.
\end{proof}

For the Rayleigh case ($m=1$), the above CDF reduces to the one in \cite{KIT}. Theorem \ref{thm1} is general and applies to any FA topology, including the ones in Fig. \ref{fig1}. It is important to point out that the performance of the FA improves as the correlation coefficients get smaller. To better show this behavior, we take an asymptotic approach and derive the achieved diversity order and outage gain \cite{TSE}. Firstly, we provide a series representation of the CDF to assist with the analysis. For the sake of convenience, we will denote
\begin{align}\label{s}
S \triangleq 1+ \sum_{n=2}^N \frac{\rho_n^2}{1-\rho_n^2}.
\end{align}

\begin{proposition}\label{pro1}
A series representation of the CDF of the estimated $\h$ can be written as
\begin{align}\label{cdf_series}
F_{\h}(x) = \frac{1}{\Gamma(m)} \sum_{k=0}^\infty \frac{c_k }{S^{m+k}}\gamma\left(m+k,Sx\right),
\end{align}
where
\begin{align}\label{ck}
c_k = \sum_{\substack{l_2,l_3,\dots,l_N \geq 0\\l_2+l_3+\cdots+l_N = k}} \alpha_{2,l_2} \alpha_{3,l_3} \cdots \alpha_{N,l_N},
\end{align}
and
\begin{align}
\alpha_{n,l} = \left(\frac{\rho_n^2}{1-\rho_n^2}\right)^l \frac{\gamma \left(m+l,\frac{x}{1-\rho_n^2}\right)}{l!\Gamma(m+l)}.
\end{align}
\end{proposition}

\begin{proof}
See Appendix \ref{pro1_prf}.
\end{proof}

Even though Proposition \ref{pro1} provides an infinite series representation, the number of terms that are actually required to achieve a certain accuracy can be finite. Now, for $\rho_n = 0$, $\forall n$, \eqref{cdf_series} reduces to the independent case as
\begin{align}
F_{\h}(x) = \left(\frac{1}{\Gamma(m)} \gamma\left(m,x\right)\right)^N,
\end{align}
since $S=1$ and all $k>0$ terms are equal to zero. This scenario can be realized with a large enough $W$ and, in particular, with the wheel topology (see Fig. \ref{fig1c} above and Fig. \ref{fig2c} below). Now, the estimated signal-to-noise ratio (SNR) at the $n$-th port is given by
\begin{align}
\eh_n = \frac{P}{\nu^2} \h_n,
\end{align}
with average SNR $\ee = \E[\eta_n] = P/\nu^2$, where $\nu^2$ is the variance of the additive white Gaussian noise (AWGN). Let $\eh$ be the largest estimated SNR, i.e. $\eh = \max\{\eh_1,\dots,\eh_n\}$. Then, the outage probability can be written as
\begin{align}\label{op}
P_{\rm o}(\theta) &= \PP\{\log_2(1+\eh) < \theta\} = F_{\h}\left(\frac{m}{\bar{\eta}}(2^\theta-1)\right),
\end{align}
where $\theta$ is a pre-defined rate threshold. Therefore, from Proposition \ref{pro1}, one can easily obtain an asymptotic expression for the outage probability.

\begin{corollary}\label{cor1}
For high SNR values, the outage probability of the estimated $\eh$ is approximated by
\begin{align}
\lim_{\bar{\eta} \to \infty} P_{\rm o}(\theta) \approx \left(\frac{m}{\bar{\eta}}(2^\theta-1)\right)^{mN} \frac{\Gamma(m+1)^{-N}}{\prod_{n=2}^N (1-\rho_n^2)^m}.
\end{align}
\end{corollary}

\begin{proof}
See Appendix \ref{cor1_prf}.
\end{proof}

Now, if the outage probability of a scheme behaves like $P_o(\theta) \approx G_o \bar{\eta}^{-G_d}$ at high SNRs, then $G_o$ is said to be the outage gain and $G_d$ is the scheme's diversity order given by \cite{TSE}
\begin{align*}
G_d = -\lim_{\bar{\eta}\to\infty} \frac{\log(P_o(\theta))}{\log(\bar{\eta})}.
\end{align*}
Therefore, based on Corollary \ref{cor1}, the diversity order of the considered FA system is $G_d = mN$ and the outage gain is equal to
\begin{align}\label{cg}
G_o = \frac{(m(2^\theta-1))^{mN}}{\Gamma(m+1)^N\prod_{n=2}^N (1-\rho_n^2)^m}.
\end{align}
It is well-known that the diversity defines the slope of the outage probability's curve and the outage gain represents the ``distance'' to the vertical axis \cite{TSE}, i.e., the smaller the value $G_o$, the better. From \eqref{cg}, we can observe that correlation negatively affects the outage gain. It follows that the achieved diversity is independent of the FA's topology but the employed topological space characterizes the achieved outage gain. As such, the minimum outage gain is obtained by the independent case, that is, when $\rho_n = 0$, $\forall n$.

\subsection{Post-scheduling Channels}
In what follows, we turn our attention to the outdated scenario. Thus, we focus on the performance of the post-scheduling $h$ conditioned on the pre-scheduling estimation $\h$. In this case, the instantaneous channels at the $n$-th port can be written as \cite{DM}
\begin{align}
g_{n,k} &= \sqrt{1-\mu_n^2} q_{n,k} + \mu_n\hat{g}_{n,k},
\end{align}
where $\hat{g}_{1,k} = \hat{x}_{1,k} + \jmath \hat{y}_{1,k}$ and $\hat{g}_{n,k} = \sqrt{1-\rho_n^2} (\hat{x}_{n,k} + \jmath \hat{y}_{n,k}) + \rho_n \hat{g}_{1,k}$ are the estimated channels with $\hat{x}_{n,k}, \hat{y}_{n,k} \sim \mathcal{N}(0,\sigma^2/2)$ and $q_{n,k} \sim \mathcal{CN}(0,\sigma^2)$ for $k = 1,\dots, m$ and $n = 1,\dots,N$; as before, we assume $\sigma = 1$. As such, the error between the pre-scheduling estimate $\h_n$ and the actual post-scheduling $h_n$ is captured by the correlation parameter $\mu_n = J_0(2\pi f T_n)$ \cite{JLV}, where $f$ is the Doppler frequency and $T_n$ is the delay between the estimation and the activation of the $n$-th port. Obviously, the delay $T_n$ is proportional to the time it needs to estimate the channel at a port \cite{JLV}, the size of the topology, the liquid's chemical properties but also the efficiency of the employed pump mechanism \cite{HUANG,PARACHA}. 

Now, let $E_n$ denotes the event that the $n$-th port has been activated, that is, the estimate at the $n$-th port is the maximum. Then, the CDF of the post-scheduling $h$ is given by
\begin{align}
F_h(x) &= \sum_{n=1}^N \int_0^\infty F_{h_n | \h_n}(x | \h_n = y) f_{\h_n | E_n}(y | E_n) p_n dy\nonumber\\
&= \sum_{n=1}^N \int_0^\infty F_{h_n | \h_n}(x | \h_n = y) f_{\h_n}(y , E_n) dy,\label{cdf2}
\end{align}
where $F_{h_n | \h_n}(\cdot | \h_n)$ and $f_{\h_n | E_n}(\cdot | E_n)$ are the conditional CDF and probability distribution function (PDF), respectively, and $p_n$ is the probability of scheduling the $n$-th port. Finally, using Bayes' theorem, $f_{\h_n}(y , E_n) = f_{\h_n | E_n}(y | E_n) p_n$, where $f_{\h_n}(y , E_n)$ is the joint PDF, given in the following proposition.

\begin{proposition}\label{prop2}
The joint PDF of $\h_n$ and $E_n$ is given by
\begin{align}\label{con_pdf1}
f_{\h_1}(x , E_1) = \frac{1}{\Gamma(m)} \exp (-x) x^{m-1} \prod_{k=2}^N\phi_k(x,x) dz,
\end{align}
for $n=1$, and
\begin{align}\label{con_pdf2}
&f_{\h_n}(x , E_n) = \frac{1}{\Gamma(m) (1-\rho_n^2)} \int_0^x \exp \left(-\frac{x+z}{1-\rho_n^2}\right)\nonumber\\
&\times\left(\frac{x z}{\rho_n^2}\right)^{\frac{m-1}{2}} I_{m-1}\left(\frac{2 \sqrt{x z \rho_n^2}}{1-\rho_n^2}\right) \prod_{\substack{k=2\\k\neq n}}^N\phi_k(z,x) dz,
\end{align}
for $n = 2, 3, \dots, N$.
\end{proposition}

\begin{proof}
See Appendix \ref{prop2_prf}.
\end{proof}

\begin{figure*}[t]\centering
	\subfloat[Linear Topology.] {\includegraphics[height=4cm]{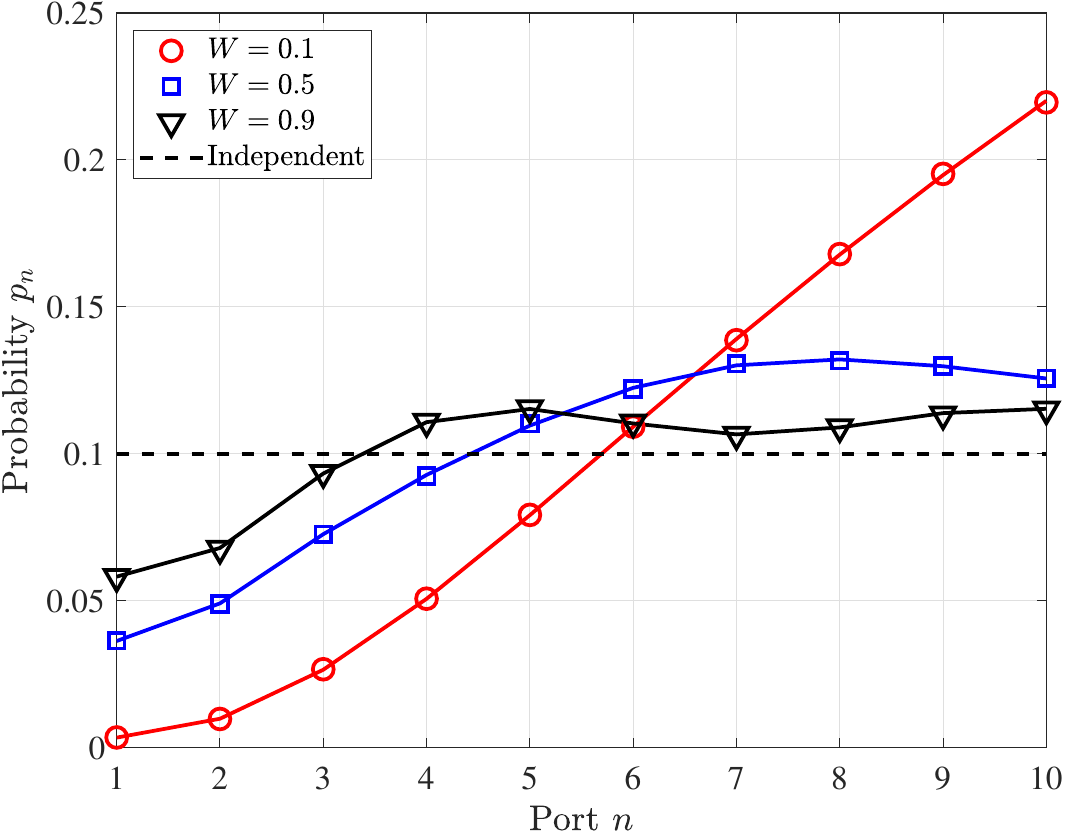}\label{fig2a}}\hspace{2mm}
	\subfloat[Circular Topology.] {\includegraphics[height=4cm]{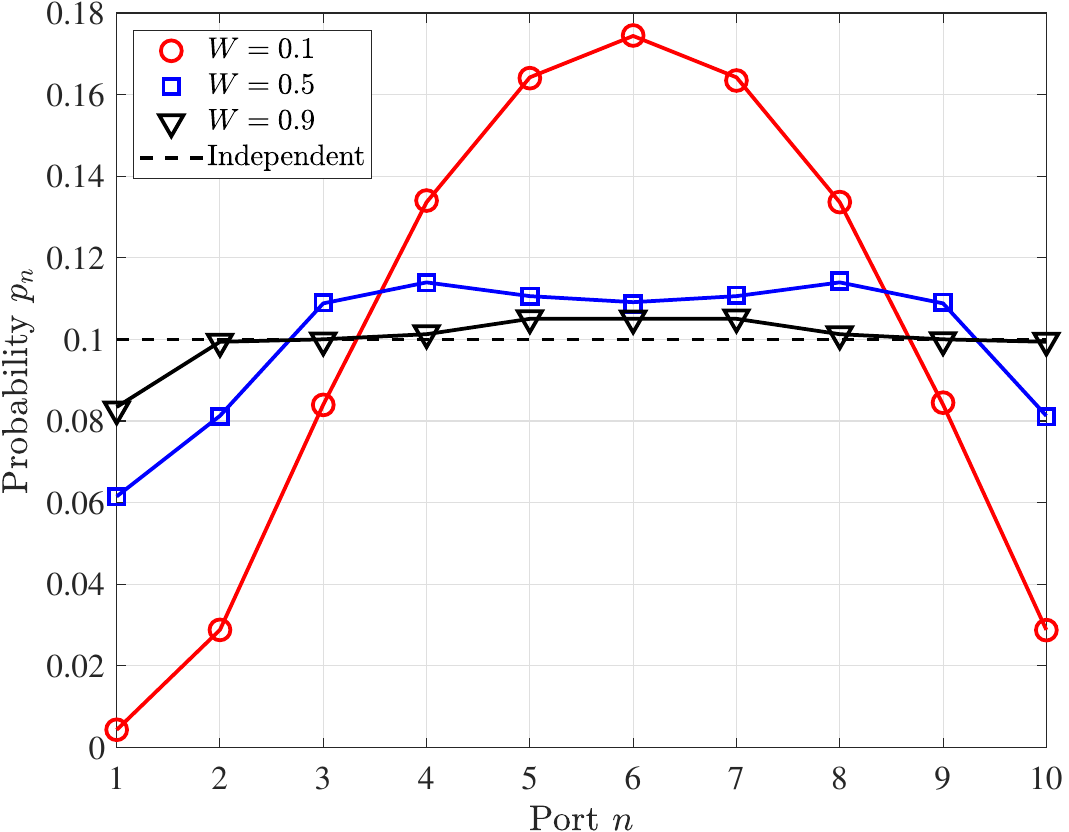}\label{fig2b}}\hspace{2mm}
	\subfloat[Wheel Topology.] {\includegraphics[height=4cm]{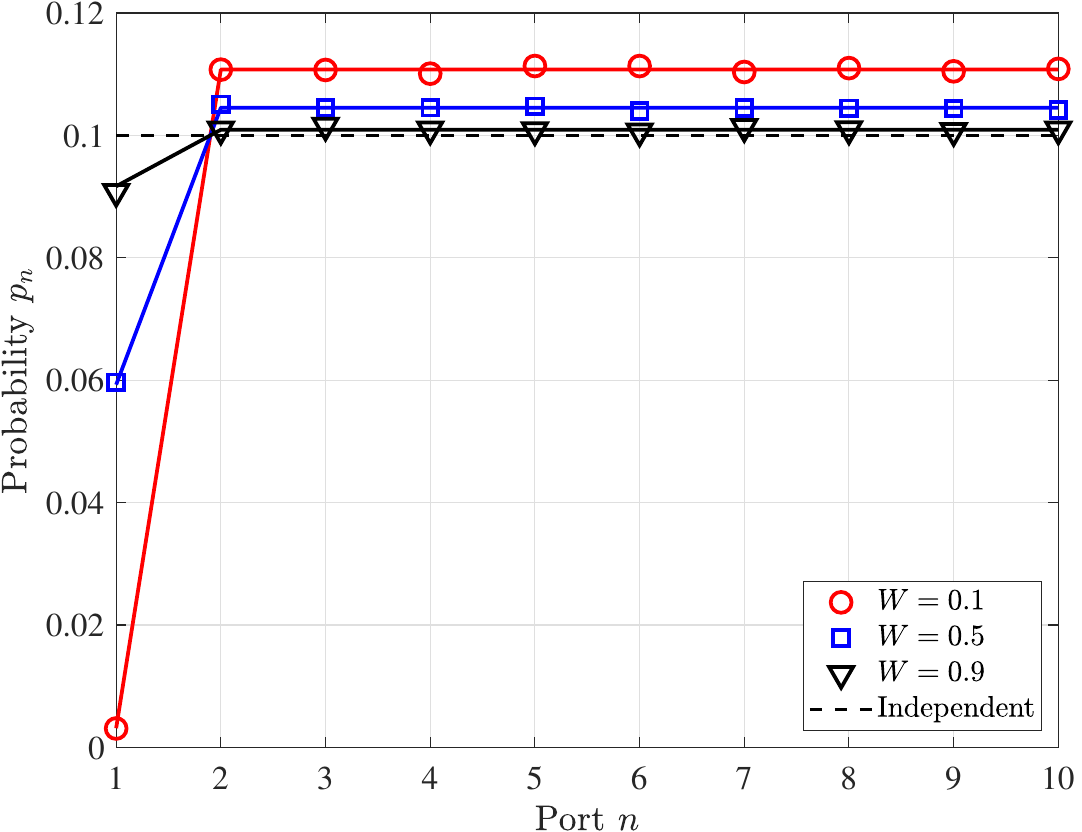}\label{fig2c}}
	\caption{Probability of selecting the $n$-th port; $N=10$.}\label{fig2}
\end{figure*}

The result in Proposition \ref{prop2} quantifies the effect of each port on the FA's overall performance. Indeed, one can easily evaluate the CDF of the estimated $\h$ by taking the sum of the joint CDFs, that is,
\begin{align}
F_{\h}(x) = \sum_{n=1}^N F_{\h_n}(x , E_n),\label{sum_con_cdf}
\end{align}
where $F_{\h_n}(x , E_n)$ are given in Appendix \ref{prop2_prf}; so \eqref{sum_con_cdf} corresponds to Theorem \ref{thm1}, albeit in a more complex form. Moreover, from Proposition \ref{prop2}, we can derive the probability of scheduling the $n$-th port, i.e.
\begin{align}
p_n = \int_0^\infty f_{\h_n}(x , E_n) dx.
\end{align}
These probabilities are illustrated in Fig. \ref{fig2}. Observe that for $\rho_n = 0$, $\forall n$, we have $Q_m(0,z) = \Gamma(m,z^2/2)/\Gamma(m)$ and the above expression gives $p_1 = p_2 = \cdots = 1/N$, which is true for the independent case. The non-uniform structure of the considered correlation model results in different selection probabilities for each port. This allows room for exploitation since one can focus on the ports that get selected more frequently in order to reduce both the delays and the complexity. Moreover, through correlation one can predict the channels at specific ports based on the channel knowledge at other ports. We can now state the final result.

\begin{theorem}\label{thm2}
The CDF of the post-scheduling $h$ with outdated channel estimates can be written as
\begin{align}\label{cdf3}
F_h(x) = 1-\sum_{n=1}^N \int_0^\infty &Q_m\left(\sqrt{\frac{2 y \mu_n^2}{1-\mu_n^2}},\sqrt{\frac{2 x}{1-\mu_n^2}}\right)\nonumber\\
&\times f_y(y , E_n) dy,
\end{align}
where $f_y(y,E_n)$ is given by Proposition \ref{prop2}.
\end{theorem}

\begin{proof}
By substituting the expressions of Proposition \ref{prop2} in \eqref{cdf2} and since $h_n \,|\, \h_n$ is a non-central chi-square random variable with $2m$ degrees of freedom and non-centrality parameter $2\h_n\mu_n^2/(1-\mu_n^2)$, the result follows. 
\end{proof}

Observe that the above expression is valid for $0 \leq \mu_n < 1$. If for a specific $n$ we have $\mu_n = 1$ (no delays), then the integral for the $n$-th term in \eqref{cdf3} is reduced to \eqref{ccdf1} or \eqref{ccdf2}, accordingly. Obviously, if $\mu_n = 1$ $\forall n$, the performance is given by Theorem \ref{thm1} (or by \eqref{sum_con_cdf}). On the other hand, if $\mu_n = 0$, $\forall n$, i.e. when the estimates are completely outdated (independent), $F_h(x)$ gives the performance of a randomly selected port, that is, $F_h(x) = 1 - \Gamma(m,x)/\Gamma(m)$, since $Q_m(0,z) = \Gamma(m,z^2/2)/\Gamma(m)$. We should also remark that a series representation of Theorem \ref{thm2} could be derived by following the same methodology as in Proposition \ref{pro1}, but we omit it for the sake of brevity. Nevertheless, we will provide a simplified asymptotic expression for high SNRs. Let $\eta$ denote the post-scheduling SNR, so the outage probability is $P_{\rm o}(\theta) = F_h\big(\frac{m}{\bar{\eta}}(2^\theta-1)\big)$.

\begin{corollary}\label{cor2}
For high SNR values, the outage probability for the post-scheduling $\eta$ simplifies to
\begin{align}
\lim_{\bar{\eta} \to \infty} P_{\rm o}(\theta) \approx &\frac{\left(\frac{m (2^\theta-1)}{\ee}\right)^m m \Gamma(mN)}{\Gamma(m+1)^{N+1} \prod _{n=2}^N (1-\rho_n^2)^m}\nonumber\\
&\times \bigg(\frac{(1-\mu_1^2)^{mN-m}}{\left(\mu_1^2+(1-\mu_1^2) S\right)^{mN}}\nonumber\\
&+\sum_{n=2}^N (1-\mu_n^2)^{mN-m} \left(\frac{1-\rho_n^2}{1-\mu_n^2 \rho_n^2}\right)^{m N}\bigg).
\end{align}
\end{corollary}

\begin{proof}
See Appendix \ref{cor2_prf}.
\end{proof}

As expected, when the channel estimates are outdated, the system is dispossessed of the channel diversity and the achieved diversity is reduced to $m$. Attaining full channel diversity would require $\mu_n \to 1$, which could be realized with extremely small FAs, i.e. $W \to 0$. Still, this may be impractical and so, in most cases, post-scheduling error due to delays is an inherent characteristic of FAs. Therefore, in what follows, we consider a linear channel prediction scheme in order to overcome this limitation \cite{KAY}.

\subsection{Linear Prediction Scheme}
We assume that the FA receiver obtains a sequence of channel estimates for each port during a training phase of duration $D$ resource blocks. In other words, the receiver obtains the following for the $n$-th port
\begin{align}
\hat{\qg}_{n,k} = [\hat{g}_{n,k}(t-(D-1)),\dots,\hat{g}_{n,k}(t-1),\hat{g}_{n,k}(t)]^\top,
\end{align}
where $\hat{g}_{n,k}(t)$ is the estimated channel at the $t$-th block for $k = 1,\dots,m$. As before, we assume perfect channel estimation at the receiver. The channel estimates for a fixed resource block are spatially correlated over different ports based on the autocorrelation function defined in \eqref{correlation}. Moreover, for fixed $n$ and $k$, the channel estimates over different resource blocks are temporally correlated, i.e. $\E\{\hat{g}_{n,k}(t_1)\hat{g}^*_{n,k}(t_2)\} = J_0(2\pi f |t_1-t_2| \tau_d),$ where $\tau_d$ is the duration of each resource block \cite{SP}. Now the aim is to predict the channel $\tilde{g}_{n,k}(t+l)$, $l$ steps ahead, for integer $l \geq 1$ based on $\hat{\qg}_{n,k}$, and so we can write \cite{KAY}
\begin{align}
\tilde{g}_{n,k}(t+l) = \qa_{n,k}^\dag \hat{\qg}_{n,k},
\end{align}
where $\qa_{n,k} \in \mathbb{C}^{1\times D}$ is the prediction coefficient vector. Due to perfect channel estimation, the optimal $\qa$ that minimizes the mean square error (MSE) is given by
\begin{align}
\qa = \qR^{-1} \qr,
\end{align}
in which, $\qR$ is the autocorrelation matrix with entries $\qR_{ij} = J_0(2\pi f |i-j| \tau_d),$ and $\qr_i = J_0(2\pi f |l+i-1| \tau_d)$, $i,j = 1,2,\dots,D$; the prediction vector can be evaluated for any integer $l$ but its accuracy diminishes as $l$ increases. Also, remark that by using the optimal solution, the prediction vector $\qa$ is the same for all ports. Finally, the correlation coefficient between the actual and the predicted channel is given by \cite{SP}
\begin{align}
\mu_0 &= \E\{\tilde{g}_{n,k}(t+l)g^*_{n,k}(t+l)\} = \sqrt{\qr^\dag \qR^{-1} \qr},
\end{align}
irrespective of the considered port. The CDF of the predicted channel can be obtained using Theorem \ref{thm2} and by setting $\mu_n = \mu_0$, $n=1,\dots,N$.

\section{Port Combining with Coded Modulation}\label{sec_coded}
In the previous section, we have derived outage probability expressions for FA port selection, where the best port is scheduled for reception. Here, we consider coded modulation design and analysis for FA port combining, in which all the ports of the FA are used for reception. Although it results in  higher receiver complexity, this combining method is superior to port selection as it exploits all the available resources, i.e. the FA ports. Based on a bound on the diversity order, we propose a low-complexity scheme through which the channel decoder can collect some amount of diversity at no additional complexity, while the remaining diversity is collected by the demodulator. Since the focus of this section is on the coded modulation design, we ignore any delays and imperfections due to the channel estimation.

\subsection{Coded Modulation Design}
We consider the transmission of a length-$\mathcal{L}$ code over an FA block-fading channel in which the $N$ ports are activated sequentially within the FA topology. For the sake of simplicity, we will focus on the linear FA (Fig. \ref{fig1a}) and the Rayleigh fading case (i.e. $m=1$). The channel model is given by
$
\qy = \qz \qS \qG + \qw = \qx \qG + \qw,
$
where $\qy$ is the complex baseband vector of received symbols with dimensions $1 \times N$, $\qz$ is the $1 \times N$ vector of modulated quadrature amplitude modulated (QAM)\footnote{We consider QAM as it provides the best performance but the proposed design could be implemented with other modulations as well.} symbols with $2^b$ symbols, $b$ being the number of bits per symbol, and $\qS$ is the $N \times N$ space-time rotation matrix with combining factor $s \leq N$; $\qS$ is used to combine the symbols in $\qz$ onto the $1 \times N$ vector $\qx$, without affecting the overall energy, i.e $\E\{\qz^\dag \qz\} = \E\{\qx^\dag \qx\}$. The goal of the rotation is to achieve higher diversity orders at the output of the demodulator, which will be explained in the sequel. As the FA receives at a given port and at a given symbol time period, the $N \times N$ channel matrix $\qG$ is diagonal with entries $g_1, \dots, g_N$. Finally, the $1 \times N$ vector $\qw$ consists of circularly-symmetric complex AWGN components with zero mean and variance $\nu^2$.

Based on the above, the digital transmission occurs as follows. At the transmitter, $\mathcal{K}$ information bits are fed to a binary encoder producing a codeword $c \in \mathcal{C} \left(\mathcal{K},\mathcal{L}, d_H \right)$, where $\mathcal{C}$ is the codewords ensemble and $d_H$ is the minimum Hamming distance of the code; the coding rate is hence $R_c = \mathcal{K}/\mathcal{L}$. We will consider trellis-based codes, although the analysis in the sequel applies to any type of binary linear code. The $\mathcal{L}$ coded bits are then interleaved and fed by groups of $b$ bits to the QAM modulator. Modulated symbols are then combined by groups of $s$ through a space-time rotation before being transmitted. At the receiver, an iterative soft-input soft-output (SISO) demodulator/decoder is implemented that consists of the exchange of extrinsic information about coded bits before deciding on the {\em a posteriori} probabilities (APP), as shown in Fig. \ref{fig3}. Specifically, the APP QAM SISO demodulator computes the extrinsic probabilities $\xi(c_i)$ given the channel likelihoods and the {\em a priori} probabilities $\pi(c_i)$ fed back from the SISO decoder as \cite{gresset08}
\begin{align}\label{ext}
\xi(c_i) = \frac{\sum_{\qx' \in \mathcal{X}(c_i = 1)} \left[\exp(-\frac{|| {\bf y}' - {\bf x}' {\bf G}||^2}{\nu^2}) \prod_{k \neq i} \pi(c_k) \right] }{\sum_{\qx \in \mathcal{X}} \left[\exp(-\frac{||{\bf y}' - {\bf x} {\bf G}||^2}{\nu^2}) \prod_{k \neq i} \pi(c_k) \right]},
\end{align}
where $c_i$ is the $i$-th bit of codeword $c$ and $\mathcal{X}$ has cardinality $2^{sb}$. In fact, due to the diagonal structure of the FA channel matrix {\bf G}, only symbols that are space-time rotated are jointly demodulated, while unrotated symbols are demodulated individually. This considerably reduces the demodulation complexity as compared to conventional multiple-antenna demodulation, in which the cardinality of the received vector is $2^{sb N_t}$, $N_t$ being the number of transmit antennas. Across the iterations, the probabilities on the coded bits computed at both the SISO demodulator and decoder become more reliable, and near-maximum likelihood (ML) performance is achieved under ideal convergence. In the sequel, we consider that the {\em a priori} probabilities $\pi(c_k)$ fed back from the SISO decoder are perfect (the {\em genie} condition \cite{gresset08}), i.e. $\pi(c_k) \in \{0,1\}$, so that the distribution at the output of the demodulator only depends on the channel likelihood, i.e. $\exp(-||\qy'-\qx'\qG||^2/\nu^2)$.

\begin{figure}[t]\centering
\includegraphics[width=0.8\linewidth]{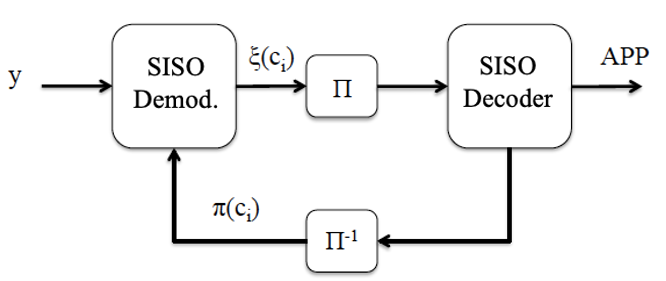}
\caption{Iterative APP demodulation and decoding receiver; $\Pi$ and $\Pi^{-1}$ represent interleaving and deinterleaving, respectively.}\label{fig3}
\end{figure}

\subsection{Diversity Order Bounds}\label{bounds-sec}
The block-fading channel with one conventional antenna at the transmitter and one FA at the receiver can be seen as a correlated block-fading channel. However, as described in Appendix \ref{thm1_prf}, by conditioning on $h_1 = |g_1|^2$, the squared magnitudes $h_n = |g_n|^2, n>1$ are independent non-central chi-square distributed random variables with two degrees of freedom.

In order to ensure code diversity and optimal coding gain, we assume that an ideal interleaver is used \cite{gresset08}, i.e., for any pair of codewords $(c,c')$, the $\omega$ non-zero bits of $c \oplus c'$ are transmitted in different fading blocks. In the iterative receiver of Fig. \ref{fig3}, the maximum channel diversity of the FA channel (i.e. equal to $N$) can be collected at the output of the demodulator, but this implies that the $N$ transmitted symbols should be space-time rotated, thus resulting in a complexity in \eqref{ext} that increases as $2^{sb}$. Thus, the lowest complexity solution for achieving full diversity $N$ would be to first collect the maximum diversity through the SISO decoder and then recover the remaining diversity at the demodulator.

\begin{proposition}\label{prop-div}
The diversity order $G_d$ of coded modulations transmitted over an $N$-port FA block-fading channel under ideal interleaving is upper-bounded as
\begin{align}\label{diversity}
G_d \leq \min \left\{\bar{s} \left\lfloor \frac{N}{\bar{s}}(1-R_c) + 1\right\rfloor; N; \lfloor \bar{s} d_H\rfloor\right\},
\end{align}
where $\bar{s}$ is the average number of rotated modulated symbols within the length-$N$ vector and the first term of the {\em min} function represents the Singleton bound on diversity.
\end{proposition}

\begin{proof}
	See Appendix \ref{proof_div}.
\end{proof}

\begin{figure}[t]\centering
\includegraphics[width=0.8\linewidth]{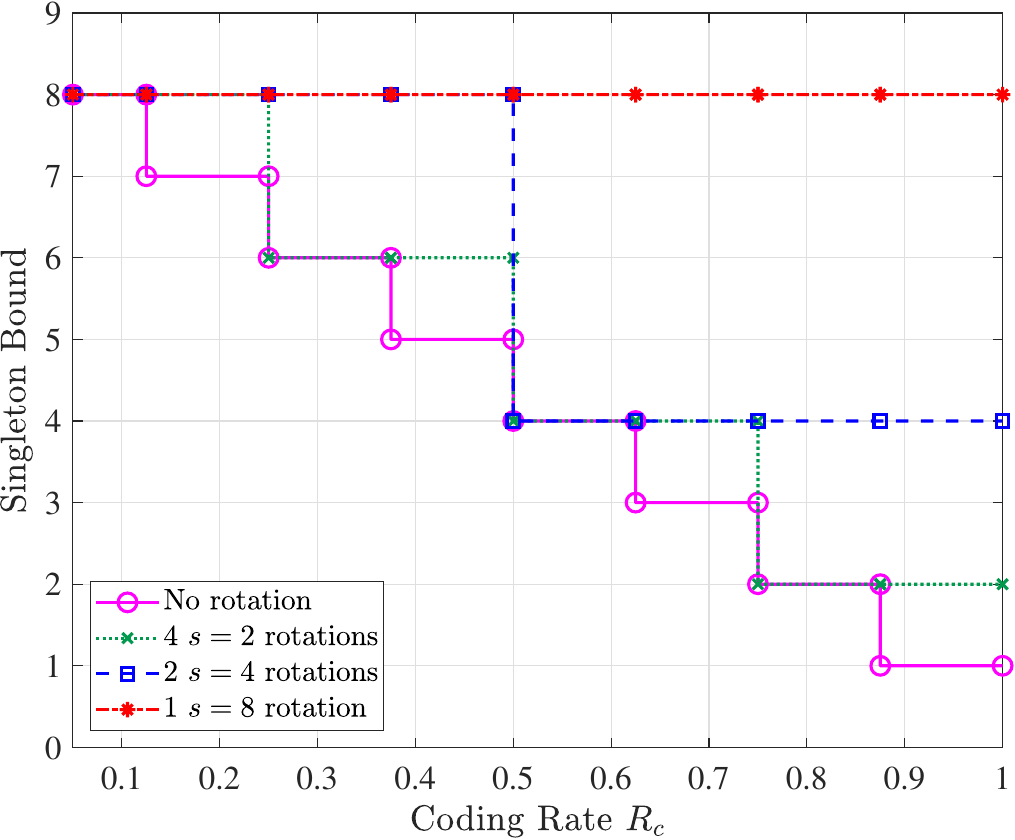}
\caption{Singleton bound on the diversity versus coding rate for an $8$-port FA block-fading channel.}\label{fig4}
\end{figure}

The bound in \eqref{diversity} provides, given a coding rate $R_c$ and maximum diversity $N$, the minimum value of $\bar{s}$ that allows the receiver to collect the maximum diversity at a minimum complexity. Fig. \ref{fig4} shows the achievable diversity of an $8$-port FA for different values of $\bar{s}$. It is shown that the larger the $\bar{s}$, the larger the coding rate that could achieve maximum diversity, however the larger the demodulator complexity will become. In case no rotation is used, the maximum rate that achieves maximum diversity is $R_c = 1/N$, and if all the symbols are rotated (i.e. $s=N$), diversity is attained with uncoded transmission (i.e. $1~s=8$ rotation). Other rotation configurations allow to achieve full diversity for $1/N < R_c \leq 1$. Examples of such rotations are given in Appendix \ref{app}.

\subsection{Average WER of Coded Modulations}\label{awer}
We now provide a tight upper bound on the average WER of space-time rotated FA block-fading channels. The WER refers to the probability that a transmitted codeword is decoded erroneously, i.e. when at least one bit is erroneous after decoding. Hence, it is a more suitable performance metric for the proposed coded modulation design than the outage probability. A general upper bound on the average WER of a linear code is derived based on the union bound as \cite{PROAKIS}
\begin{equation}\label{pew}
P_{\rm WER} \leq \sum_{c \in \mathcal{C}\backslash\{0\}} \PP\{0 \to c\},
\end{equation}
where $\PP\{0 \to c\}$ denotes the PEP, that is, the probability of decoding codeword $c$ when the all-zero codeword $\{0\}$ is transmitted. Under ideal interleaving, the bits of an error event are placed on different channel realizations to achieve code diversity, as explained above. By assuming the all-zero codeword is transmitted, we have an error if the decoder gives codeword $c$ with Hamming weight $w(c)>0$. We now set $w_k(c)$ as the partial Hamming weight that is transmitted on fading block $k$, with $\sum_{k=1}^\mathcal{B} w_k(c) = w(c)$, where $\mathcal{B}$ is the total number of fading blocks \cite{boutros04}. To achieve diversity, if $d_H \geq \mathcal{B}$ we should have $w_k(c) > 0,$ $\forall k$, and the $w(c)$ non-zero bits should be equally distributed over the $k$ fading blocks to ensure high coding gains; these two conditions are ensured by the ideal interleavers. If no rotation is used, the fading blocks are represented by the $g_k$ coefficients (i.e. $\mathcal{B}=N$). If a rotation is used, a block is defined as a group of $s$ rotated channel coefficients.

\begin{proposition}\label{prop-pep}
The conditional PEP under ML decoding of space-time rotated $N$-port FA block-fading channels with ideal interleaving, codeword Hamming weight $w(c) = i$, $2^b$-QAM modulation, and Gray mapping is approximated by \eqref{approx1},
\begin{figure*}[t!]\begin{equation}\label{approx1}
\PP\{0 \to c \,|\, \qS, \qG\} \approx \delta Q \left(\sqrt{R_c \ee \left( \zeta \sum_{j=1}^{\psi/s} w_j(c) \sum_{k=1}^s \kappa_k h_{(j-1)s + k} + \zeta' \sum_{j=\psi + 1}^N w_j(c) h_j \right)}\right).
\end{equation}\noindent\makebox[\linewidth]{\rule{18.3cm}{.4pt}}\end{figure*}
where $\ee$ is the SNR per bit, $\delta = [4/(\bar{s}b)] (\sqrt{2^{\bar{s}b}} - 1)/\sqrt{2^{\bar{s}b}}$, $\zeta = 3b/(2^b - 1)$, $\zeta' = 3sb/(2^{sb} - 1)$, $\psi$ denotes the number of symbols out of the $N$ that are space-time rotated, and the $\kappa_k$ are functions of the space-time rotation matrix entries.
\end{proposition}

\begin{proof}
See Appendix \ref{proof-pep}.
\end{proof}

The argument of the $Q$-function in \eqref{approx1} shows that space-time rotations, apart from achieving higher diversity orders with higher coding rates, increase the coding gain. As an example, consider a code $\mathcal{C}$ with $d_H=5$ transmitted on a $2$-port FA block-fading channel. Without rotation, the interleaver gives $w_1(c)=2$ that multiplies $h_1$ and $w_2(c)=3$ that multiplies $h_2$. By using a rotation with $s=2$, we get $w_1(c)=d_H=5$ that multiplies $\sum_{\ell=1}^s h_\ell$, which increases the argument of the $Q$-function, thus decreases the PEP.

\begin{corollary}\label{cor3}
For high SNR, we obtain \eqref{asymp-pep},
\begin{figure*}[t!]\begin{equation}\label{asymp-pep}
\lim_{\ee\to\infty} \PP\{0 \to c\} \leq \left(\frac{1}{R_c \ee}\right)^N \frac{\delta\zeta^{-\psi}(\zeta')^{\psi-N}}{2 w_1(c) \kappa_1 \prod_{j=\psi+1}^N (1-\rho_j^2) w_j(c) \prod_{j=1}^{\psi/s} \prod_{k=1}^s(1-\rho_{(j-1)s + k}^2) w_j(c) \kappa_k}.
\end{equation}\noindent\makebox[\linewidth]{\rule{18.3cm}{.4pt}}\end{figure*}
for $(j-1)s + k \neq 1$.
\end{corollary}

\begin{proof}
See Appendix \ref{coro-pep-proof}.
\end{proof}

The asymptotic expression of the PEP of \eqref{asymp-pep} shows that maximum coded diversity of order $N$ is attained, as the PEP is inversely proportional to $\ee^N$.

\begin{proposition}
A tight upper bound on the average WER of space-time rotated $N$-port FA block-fading channels is given by
\begin{align}
P_{\rm WER} \leq 1 - \int_\qG &[ 1 - \min\{1, W_i \PP\{0 \to c \,|\, \qS , \qG\}\}]^{\mathcal{L}} p(\qG) d\qG,
\end{align}
where $p(\qG)$ denotes the distribution of $\qG$, and $W_i$ denotes the number of codewords with Hamming weight $w(c)=i$.
\end{proposition}

\begin{proof}
The upper bound is obtained by replacing $\PP\{0 \to c \,|\, \qS , \qG\}$ from \eqref{approx1} in the upper bound on the WER expression from \cite{malka99}.
\end{proof}

It is important to point out that the value $W_i$ can be obtained in a straightforward manner for trellis codes, while it might be more difficult to obtain for block codes. Moreover, for high SNR, we can almost surely state that $W_i \lim_{\ee\to\infty} \PP\{0 \to c\} < 1$. Therefore, we can write
\begin{equation}
\lim_{\ee\to\infty} P_{\rm WER} \leq \mathcal{L} \sum_{i=d_H}^\infty W_i \PP\{0 \to c\},
\end{equation}
which results from $(1-x)^a \approx 1-ax$ for $x \approx 0$ and $\PP\{0 \to c\}$ is given by Corollary \ref{cor3}.

\section{Numerical Results}\label{numerical}
We now validate our theoretical analysis with computer simulations. We will start with the performance of the port selection (Section \ref{sec_outdated}). For the sake of presentation, we consider a simple model for $T_n$, i.e. the delay at the $n$-th port, given by
$T_n = \left(\frac{N-n+1}{N}\right) \tau_e\beta$, where $\tau_e$ is the duration for estimating the channel at a port and $\beta$ is a parameter that captures the FA's topological effects, given by $W$, $\pi W$ and $NW/2$ for the linear (length), the circular (circumference) and the wheel-shaped topology (width), respectively. Also, unless otherwise stated, we consider $\theta = 2$ bps, $m = 2$, $f = 100$ Hz \cite{SP}, $\tau_e = 1/(10f)$ s, $\tau_d = 1/(100 f)$ s \cite{SP}, and $l=1$.

\begin{figure}[t]\centering
	\includegraphics[width=0.9\linewidth]{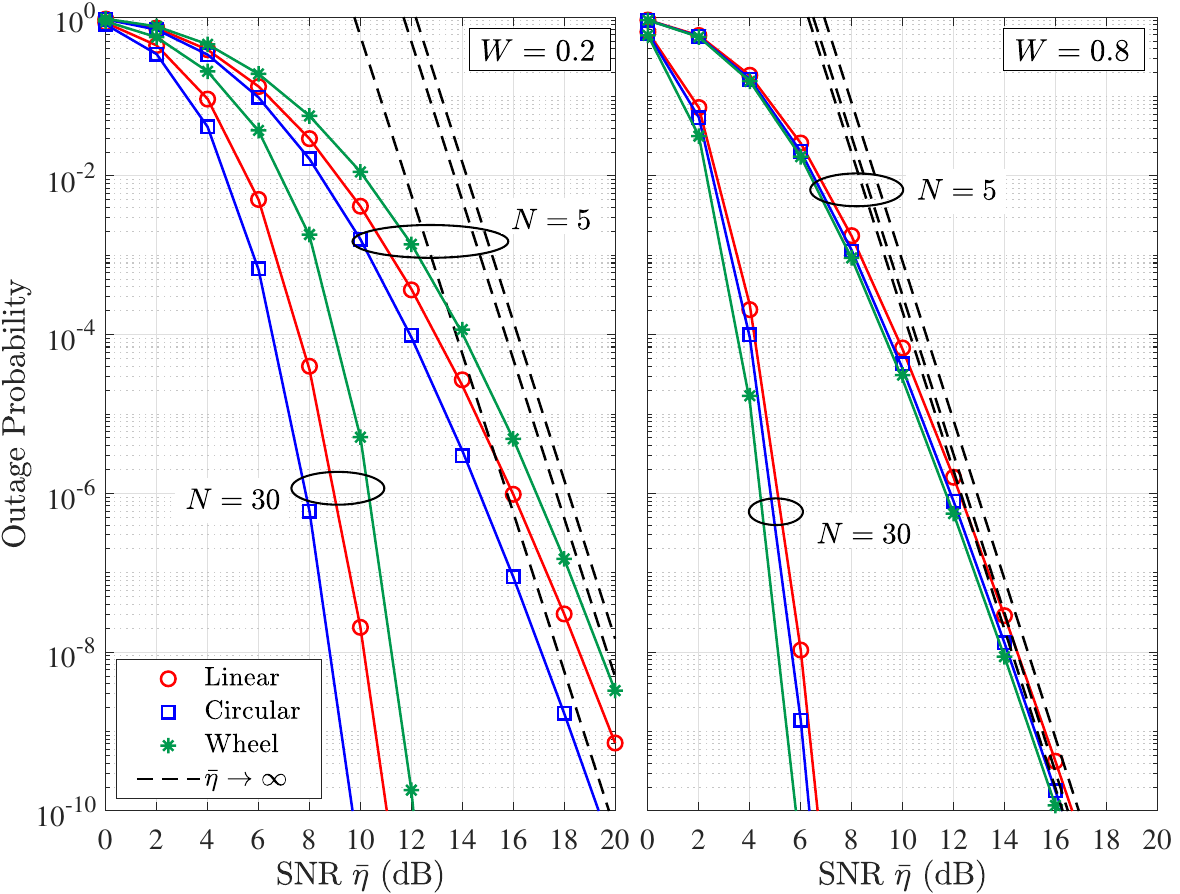}
	\caption{Outage probability of pre-scheduling channels; lines and markers correspond to theoretical and simulation results, respectively.}\label{fig5}
\end{figure}

Fig. \ref{fig5} illustrates the outage probability achieved by the three topologies in terms of the SNR for $N=5, 30$ and $W=0.2,0.8$. As shown in Corollary \ref{cor1}, the FA realizes full diversity gains irrespective of the topology's size and shape. On the other hand, the ``shift'' towards the $y$-axis (i.e. the outage gain), differs due to the fact that it depends on the spatial correlation. Specifically, for a small FA size ($W=0.2$), the spatial correlation is relatively larger. The circular topology outperforms the other two as its circular shape provides longer distances between ports. The wheel topology however has the worst performance, since the correlation between the first and any other port remains equal. The linear architecture achieves a balance between the two. These observations change for a larger FA size ($W=0.8$). Indeed, the wheel topology has the best performance albeit by a small margin. However, as can be seen, this margin increases with $N$. Finally, our analytical expressions (lines) perfectly match the simulation results (markers), which validates our theoretical methodology.

\begin{figure}[t]\centering
\includegraphics[width=0.9\linewidth]{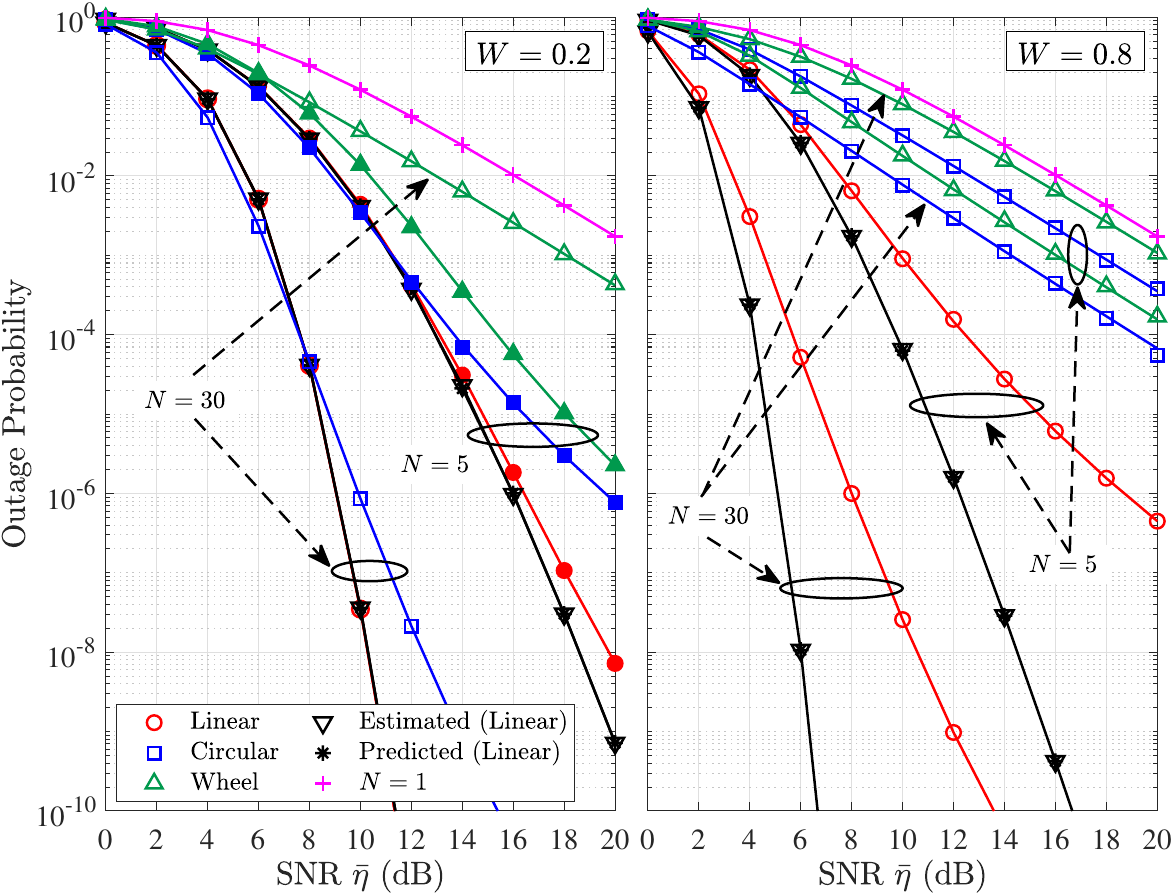}
\caption{Outage probability of post-scheduling and predicted channels; lines and markers correspond to theoretical and simulation results, respectively.}\label{fig6}
\end{figure}

In Fig. \ref{fig6}, the outage probability versus the SNR is depicted for the outdated scenario. For comparison purposes, we also illustrate the performance of the pre-scheduling estimated channels and the case $N=1$. It can be observed that the remarks stated for Fig. \ref{fig5} are no longer valid. Firstly, the channel diversity has been lost and the achieved diversity here is $m=2$. Moreover, the linear architecture generally outperforms the other two in all scenarios as it provides the lowest delays between estimating and activating a port. Essentially, the linear architecture, due to its simplified structure, provides a good balance between spatial correlation and delays. In contrast, the liquid displacement in the other two architectures is subject to longer delays. This also explains the fact that a smaller number of ports achieves a better performance with the wheel-shaped topology. In Fig. \ref{fig6}, we also show the performance achieved by the prediction scheme using the linear topology with $D=4$ and $l=1$. It is clear that the prediction scheme performs as well as the estimated case and thus attains full channel diversity. As before, the analysis (lines) matches the simulations (markers) for the outdated case, which validates the accuracy of our analysis.

\begin{figure}[t]\centering
	\includegraphics[width=0.9\linewidth]{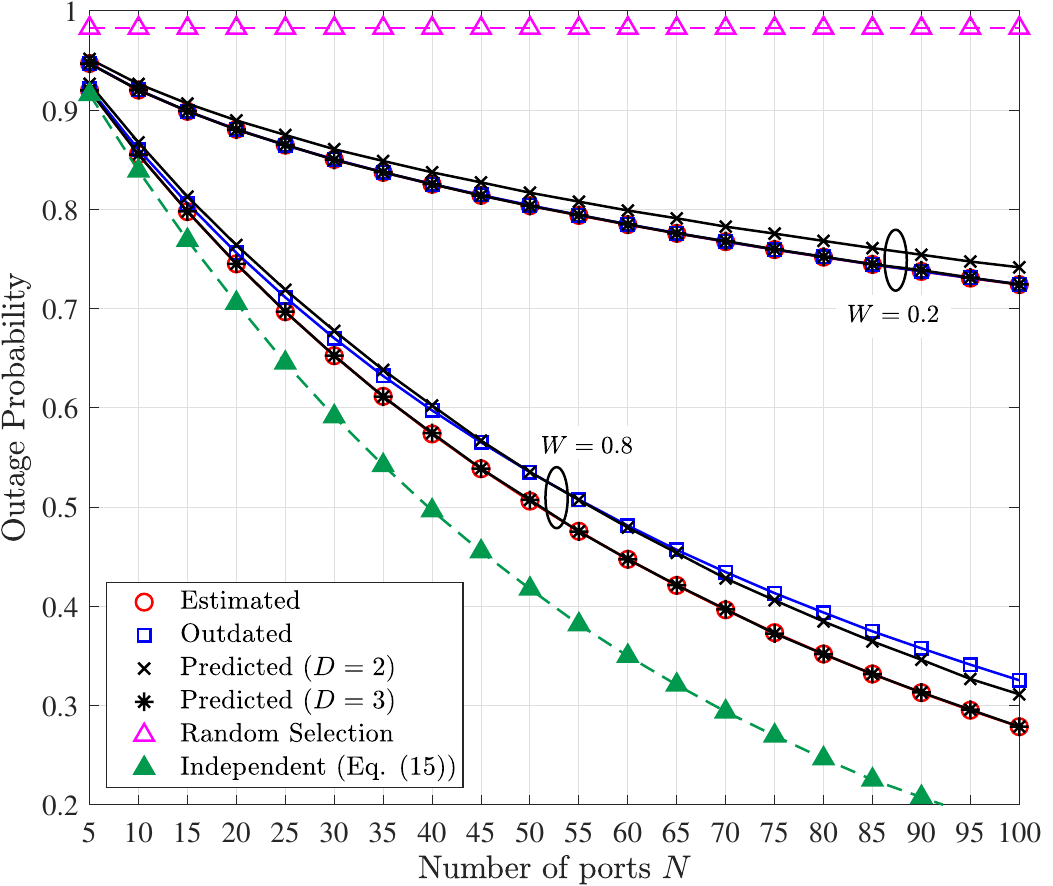}
	\caption{Outage probability versus the number of ports; lines and markers correspond to theoretical and simulation results, respectively; $\ee = 0$ dB.}\label{fig7}
\end{figure}

Fig. \ref{fig7} shows the outage performance with respect to the number of ports for the linear topology. The figure also depicts as benchmark the performance with random selection (upper bound) and the scenario where the channels are independent (lower bound). As expected, as $W$ increases, the performance gets closer to the independent case. Moreover, with $W=0.2$, the outdated channel estimates do not affect the performance. Nevertheless, with $W=0.8$, the performance gap between pre-schediling and post-scheduling channels increases with the number of ports. Finally, the prediction scheme performs very well for both values of $W$. In the case $D=2$, the performance loss is relatively small and it actually outperforms the outdated scenario for large values of $N$. On the other hand, using more resource blocks ($D=3$), attains the best possible performance.

Next, we present the performance of the coded modulation design for FAs (Section \ref{sec_coded}). Throughout the simulations, non-recursive non-systematic convolutional (NRNSC) codes with different constraint length $L$ and free distance $d_{\rm free}$ are used for different coding rates, as shown in Table \ref{table:1} \cite{frenger}. The codeword length is $\mathcal{L}=1024$ bits and optimal interleavers from \cite{gresset08} are implemented. Moreover, SISO decoding of the convolutional code is performed through the ``Forward-Backward'' algorithm \cite{bcjr}. The space-time rotations employed in the simulations are given in Appendix \ref{app} for different sizes and value of the combining factor $s$. Finally, we use as a benchmark the Gaussian input outage probability of the $N$-port FA block-fading channel given as
\begin{align}
&\PP\{\log_2 \det(\qI_N + \ee \qG \qG^\dag) < N\theta\}\nonumber\\
&= \PP\{\det(\qI_N + \ee \qG \qG^\dag) < 2^{N\theta}\},\label{pout}
\end{align}
where $\theta = b R_c$ is the target information rate during one channel use and $\ee$ is the SNR per symbol; note that the outage probability provides the lower bound on the performance of coded transmission.

\begin{figure}[t]
\begin{minipage}{.49\textwidth}\centering
\includegraphics[width=\textwidth]{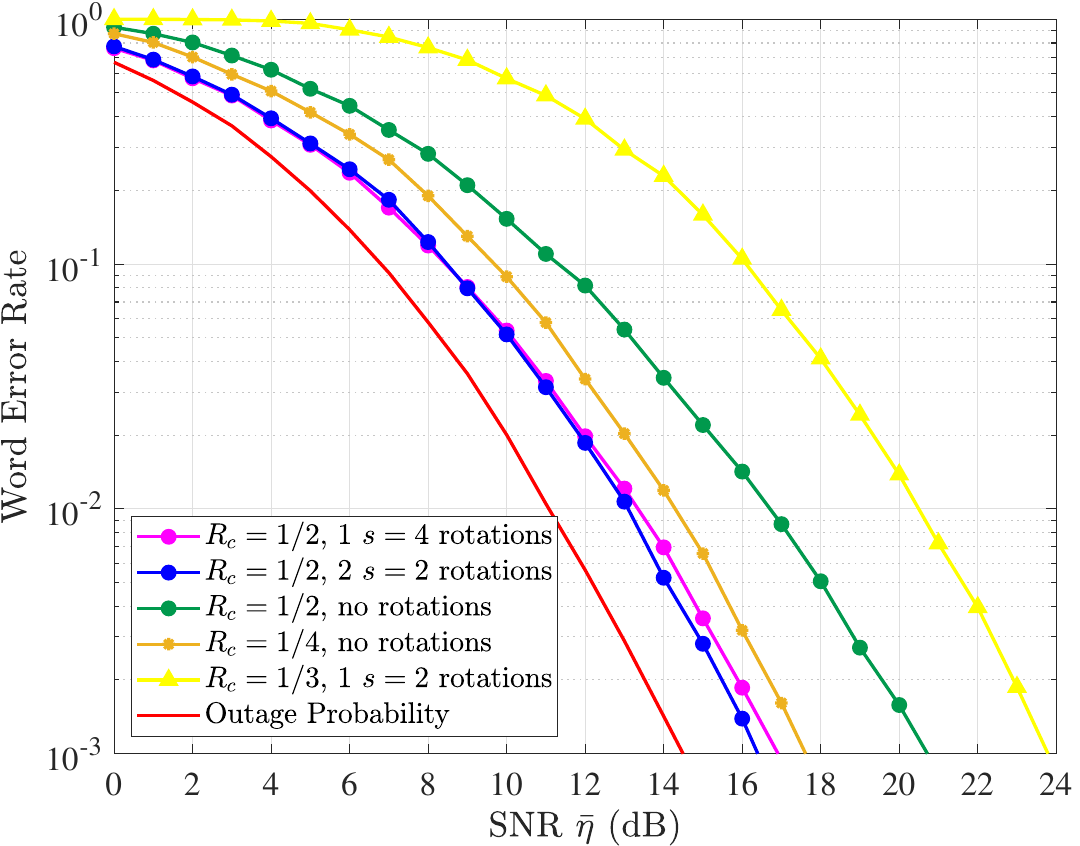}
\caption{WER performance for different coding rates.}\label{fig8}
\end{minipage}\hfill
\begin{minipage}{.49\textwidth}\centering
\includegraphics[width=\textwidth]{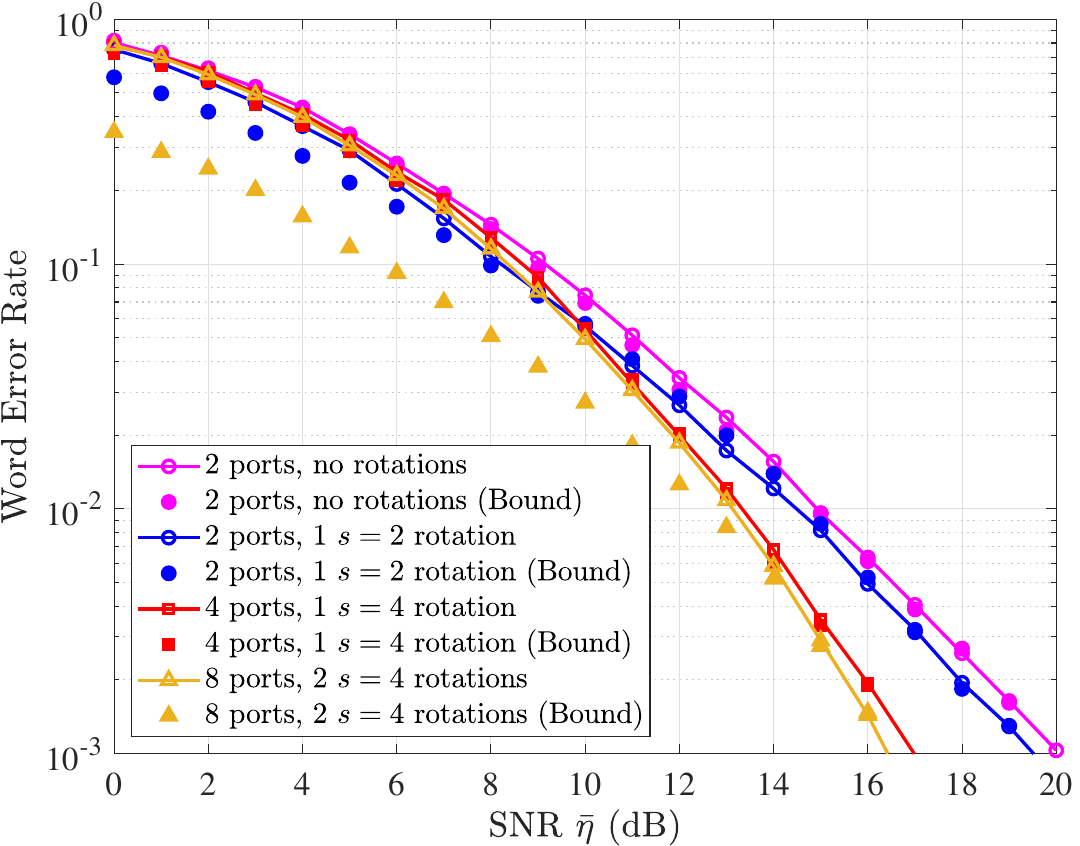}
\caption{WER of the $R_c=1/2$ $(133,171)_8$ NRNSC code, BPSK modulation.}\label{fig9}
\end{minipage}
\end{figure}

\begin{table}[t]
\caption{NRNSC codes used in the simulations.}\centering
\begin{tabular}{|c |c |c | c|}\hline
$L$ & $R_c$ & Generator (base 8) & $d_{\rm free}$\\[0.5ex]\hline\hline
$7$ &1/2 & $(133,171)$ & 10\\\hline
$5$ & 1/3 & $(25,33,37)$ & 12\\\hline
$4$ & 1/4 & $(13,15,15,17)$ & 13\\[0.5ex]\hline
\end{tabular}\label{table:1}
\end{table}

\begin{figure}[t]\centering
\includegraphics[width=0.9\linewidth]{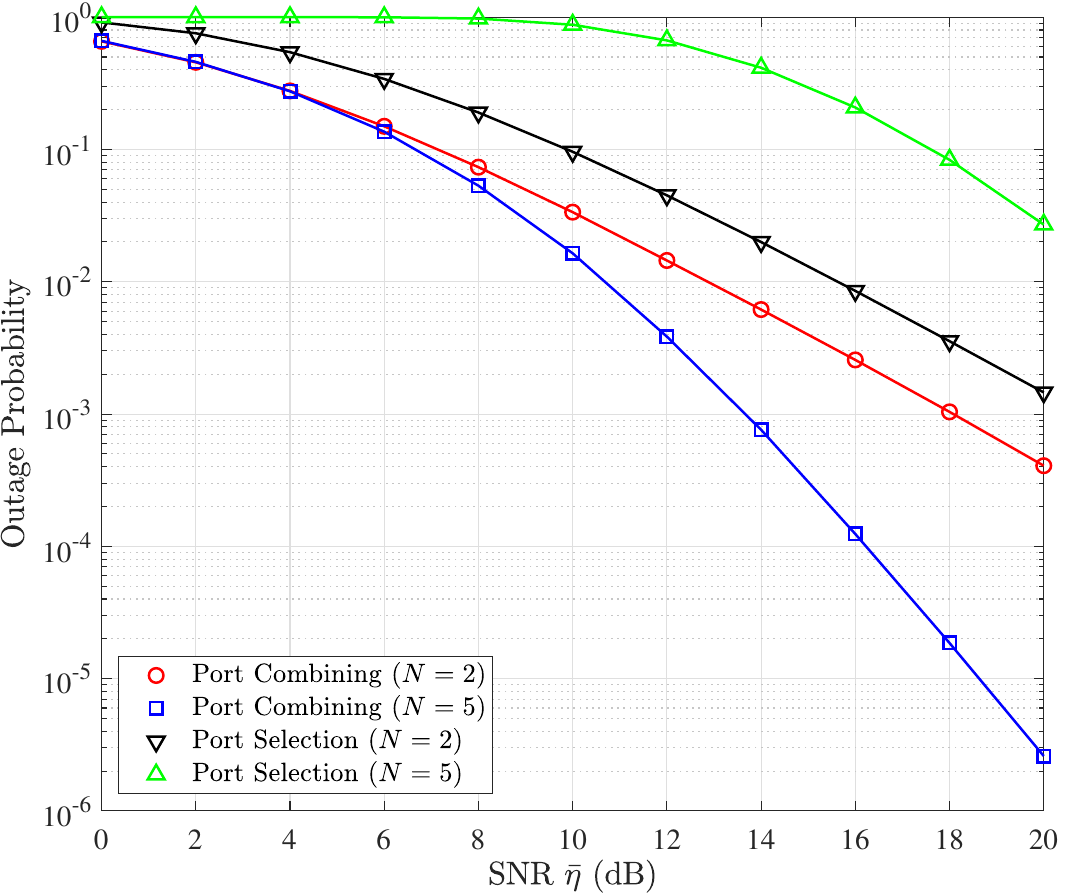}
\caption{Port selection vs port combining.}\label{fig10}
\end{figure}

In Fig. \ref{fig8}, results for a $4$-port receive FA are shown for a transmission rate of $R=1$ bit/s/Hz, which means that a QAM with $b$ bits per symbol is used with a code of rate $R_c = 1/b$. This justifies using codes with increasing free distances (see Table \ref{table:1}) in a goal to compensate for the loss in performance resulting from higher order modulations. As shown in the green curve, the rate-$1/2$ code achieves a diversity $G_d = 3 < N$ without rotation, as given in \eqref{diversity}. When adding two space-time rotations with $s=2$ or one rotation with $s=4$, full diversity is achieved, however a smaller $s$ induces a lower demodulation complexity. On the other hand, the rate-$1/3$ code requires only one rotation with $s=2$ to achieve full diversity, and the rate-$1/4$ code achieves $G_d=N=4$ without rotation. It is worth mentioning that the additional coding gain provided by a space-time rotation comes at the expense of more iterations needed for convergence (up to four iterations between the SISO demodulator and decoder, as compared to one in the absence of a rotation). In Fig. \ref{fig9}, WER results are shown together with the bounds on the WER from sub-section \ref{awer}. When no rotation is employed, the bound is tight for all the SNR range. In the presence of space-time rotations, the bound is tight for high SNR values. Finally, we compare the two schemes in Fig. \ref{fig10}. For a fair comparison, we set the information rate threshold for both schemes to $N\theta$, as in \eqref{pout}, so that it takes into account the total rate of the transmission; note that, in this way, the comparison is only valid between the two schemes for the same spectral efficiency. We can observe that both schemes attain the same diversity for the same number of ports but the combining scheme outperforms the selection significantly, in particular for a larger number of ports.

\section{Conclusions}\label{conclusion}
In this paper, we focused on the diversity and coded modulation design of FA systems. We firstly studied the outage probability of FA systems under general Nakagami fading channels and analytical expressions were provided for the performance with and without post-scheduling errors. It was shown that despite FAs achieving maximum channel diversity, equal to the number of ports, this diversity was dispossessed due to scheduling delays. Therefore, we proposed a linear prediction scheme, which overcomes this limitation and can achieve full diversity. Moreover, we designed space-time coded modulations that attain optimal performance over block-fading channels, by combining space-time rotations with code diversity. We analytically derived the pairwise error probability and provided tight bounds for the WER. The proposed coded modulations achieve maximum diversity and require low-complexity implementation at the receiver. Our result exhibit the potentials of FA in communication systems but also propose solutions for their limitations.

\appendix
\subsection{Proof of Theorem \ref{thm1}}\label{thm1_prf}
By using \eqref{rv1} and \eqref{rv2}, the estimated channels can be written as  $\hat{g}_{1,k} = \hat{x}_{1,k} + \jmath \hat{y}_{1,k}$ and $\hat{g}_{n,k} = \sqrt{1-\rho_n^2} (\hat{x}_{n,k} + \jmath \hat{y}_{n,k}) + \rho_n \hat{g}_{1,k}$. Now, by conditioning on both $\hat{x}_{1,k}$ and $\hat{y}_{1,k}$ and then scaling the variances of $\hat{x}_{n,k}$ and $\hat{y}_{n,k}$ to one, we can write
\begin{align}
\hat{g}_{n,k} = \sqrt{\frac{1-\rho_n^2}{2}} (\tilde{x}_{n,k} + \jmath \tilde{y}_{n,k}),
\end{align}
where $\tilde{x}_{n,k} \,\big|\, \hat{x}_{1,k} \sim \mathcal{N}\left(\sqrt{\frac{2}{1-\rho_n^2}} \rho_n \hat{x}_{1,k}, 1\right)$ and $\tilde{y}_{n,k} \,\big|\, \hat{y}_{1,k} \sim \mathcal{N}\left(\sqrt{\frac{2}{1-\rho_n^2}} \rho_n \hat{y}_{1,k}, 1\right)$. From \eqref{sum}, the estimated $\h_n$ is given by
\begin{align}
\h_n = \sum_{k=1}^m |\hat{g}_{n,k}|^2 &= \frac{1-\rho_n^2}{2}\sum_{k=1}^m (\tilde{x}_{n,k}^2 + \tilde{y}_{n,k}^2)\nonumber\\
&\triangleq \frac{1-\rho_n^2}{2} \tilde{h}_n,
\end{align}
where, due to the condition on $\hat{x}_{1,k}$ and $\hat{y}_{1,k}$, $\tilde{h}_n$ are mutually independent non-central chi-square random variables with $2m$ degrees of freedom and non-centrality parameter $\frac{2z\rho_n^2}{1-\rho_n^2}$ with $z \triangleq \sum_{k=1}^m (\hat{x}_{1,k}^2 + \hat{y}_{1,k}^2)$. Therefore, the CDF of $\h_n$ given $z$, denoted by $\phi_n(z,x)$, is
\begin{align}
\phi_n(z,x) &= F_{\h_n|z}(x|z) = \PP\{\h_n < x | z\}\nonumber\\
&= 1-Q_m\left(\sqrt{\frac{2 z \rho_n^2}{1-\rho_n^2}},\sqrt{\frac{2x}{1-\rho_n^2}}\right).\label{cdf}
\end{align}
Now, since $\h_2, \h_3, \dots, \h_N$ are independent, their joint CDF given $z$ is written as
\begin{align}
F_{\h_2,\h_3,\dots,\h_N|z}(x|z) = \prod_{n=2}^N \phi_n(z,x),
\end{align}
where $\phi_n(z,x)$ is given by \eqref{cdf}. Then, the CDF of the maximum estimated $\h$ is given by
\begin{align}
F_{\h}(x) &= \E_z \left\{F_{\h_2,\h_3,\dots,\h_N|z}(x|z)\right\}\nonumber\\
&= \frac{1}{\Gamma (m)} \int_0^x \exp (-z) z^{m-1} \prod_{n=2}^N \phi_n(z,x) dz,\label{int_form}
\end{align}
which follows from the fact that $z$ is a central chi-square random variable with $2m$ degrees of freedom.

\subsection{Proof of Proposition \ref{pro1}}\label{pro1_prf}
In order to simplify \eqref{cdf1}, we use the following series representation of the Marcum-$Q$ function %\cite{XX}
\begin{align}\label{series}
Q_m(a,b) = 1 - \exp\left(-\frac{a^2}{2}\right) \sum_{k=0}^\infty \frac{\gamma(m+k,b^2/2)}{k!\Gamma(m+k)} \left(\frac{a^2}{2}\right)^k.
\end{align}
Therefore, we can write
\begin{align}
F_{\h}(x) &= \frac{1}{\Gamma (m)} \int_0^x \exp (-z) z^{m-1} \prod_{n=2}^N \exp\left(-\frac{z \rho_n^2}{1-\rho_n^2}\right)\nonumber\\
&\qquad\times \sum_{k=0}^\infty \frac{\gamma\left(m+k,\frac{x}{1-\rho_n^2}\right)}{k!\Gamma(m+k)} \left(\frac{z \rho_n^2}{1-\rho_n^2}\right)^k dz\nonumber\\
&= \frac{1}{\Gamma(m)} \int_0^x \exp\left(-Sz\right) z^{m-1}\nonumber\\
&\qquad\times \prod_{n=2}^N \sum_{k=0}^\infty \frac{\gamma\left(m+k,\frac{x}{1-\rho_n^2}\right)}{k!\Gamma(m+k)} \left(\frac{z \rho_n^2}{1-\rho_n^2}\right)^k dz,
\end{align}
where $S$ has been defined in \eqref{s}. The above expression involves the Cauchy product of $N-1$ power series. Hence, it follows that
\begin{align}
F_{\h}(x) &= \frac{1}{\Gamma (m)} \int_0^x \exp\left(-Sz\right) z^{m-1} \sum_{k=0}^\infty c_k z^k dz\nonumber\\
&= \frac{1}{\Gamma (m)} \sum_{k=0}^\infty c_k \int_0^x \exp\left(-Sz\right) z^{k+m-1} dz,
\end{align}
where the coefficients $c_k$ are given by \eqref{ck}. Finally, the proposition is proven by using the transformation $z \to S/t$ and the fact that $\int_0^b \exp(-t)t^{a-1} dt = \gamma(a,b)$ \cite{GRAD}.

\subsection{Proof of Corollary \ref{cor1}}\label{cor1_prf}
By using \eqref{cdf_series}, we can write the outage probability, defined in \eqref{op}, as
\begin{align}
P_{\rm o}(\theta) = \frac{1}{\Gamma(m)} \sum_{k=0}^\infty \frac{c_k }{S^{m+k}}\gamma\left(m+k,\frac{m}{\ee}(2^\theta-1)S\right).
\end{align}
Thus, for $\ee \to \infty$,
\begin{align}
&\lim_{\ee \to \infty} P_{\rm o}(\theta) \to \frac{1}{\Gamma(m)}\sum_{k=0}^\infty\frac{c_k \left(\frac{m}{\ee}(2^\theta-1)S\right)^{m+k}}{(m+k)S^{m+k}},
\end{align}
which follows from the fact that $\gamma(a,b) \to b^a/a$ for $x \to 0$. 
In this case, the term $k=0$ dominates and so
\begin{align}
\lim_{\ee \to \infty} P_{\rm o}(\theta) \to \frac{1}{\Gamma(m)}\frac{c_0 }{m}\left(\frac{m}{\ee}(2^\theta-1)\right)^m,
\end{align}
with
\begin{align}
c_0 = \prod_{n=2}^N\frac{1}{m\Gamma(m)}\left(\frac{m(2^\theta-1)}{\ee (1-\rho_n^2)}\right)^m.
\end{align}
As $m\Gamma(m) = \Gamma(m+1)$, the result follows after several algebraic manipulations.

\subsection{Proof of Proposition \ref{prop2}}\label{prop2_prf}
We start by deriving the joint CDF $F_{\h_n}(x, E_n)$. Now, the CDF when the first port is selected is given by
\begin{align}
F_{\h_1}(x , E_1) &= \PP\{h_1 < x , h_1 > \max\{h_2,\dots,h_N\}\}\nonumber\\
&= \PP\{h_1 < x , h_1 > h_2, \dots, h_1 > h_N\}.
\end{align}
By fixing $h_1$, the events above are independent and so we can write
\begin{align}
F_{\h_1}(x , E_1) &= \E_{h_1}\left\{\prod_{n=2}^N \phi_n(h_1,h_1)\,\bigg|\, h_1 < x\right\}\nonumber\\
&= \frac{1}{\Gamma(m)} \int_0^x \exp(-z) z^{m-1} \prod_{k=2}^N\phi_k(z,z) dz,\label{ccdf1}
\end{align}
which follows the same approach as in Appendix \ref{thm1_prf}. Now, for the $n$-th port, $1 < n \leq N$, we have
\begin{align}
F_{\h_n}(x , E_n) = \PP\{&h_n < x , h_n > h_1, \dots, h_n > h_{n-1},\nonumber\\& h_n > h_{n+1}, \dots, h_n > h_N\}.
\end{align}
In this case, conditioning on both $h_1$ and $h_n$, we get
\begin{align}
F_{\h_n}(x , E_n) &= \E_{h_1,h_n}\left\{\prod_{\substack{k=2\\k\neq n}}^N\phi_k(h_1,h_n) \Bigg| h_1 < h_n < x \right\}\nonumber\\
&= \int_0^x \int_0^y f_{h_1,h_n}(z,y) \prod_{\substack{k=2\\k\neq n}}^N\phi_k(z,y) dz dy,
\end{align}
where $f_{h_1,h_n}(\cdot,\cdot)$ is the joint PDF of $h_1$ and $h_n$, which can be obtained with Bayes' rule as
\begin{align}
&f_{h_1,h_n}(z,y) = f_{h_n|h_1}(z|y) f_{h_1}(y)\\
&= \frac{1}{\Gamma(m) (1-\rho_n^2)} \exp \left(-\frac{y+z}{1-\rho_n^2}\right) \left(\frac{y z}{\rho_n^2}\right)^{\frac{m-1}{2}}\nonumber\\
&\qquad\times I_{m-1}\left(\frac{2 \sqrt{y z \rho_n^2}}{1-\rho_n^2}\right),
\end{align}
where $f_{h_n|h_1}(z|y)$ is the conditional PDF of a non-central chi-square random variable and $f_{h_1}(y)$ is the PDF of a central chi-square random, both of $2m$ degrees of freedom. Then, we can write
\begin{align}
&F_{\h_n}(x , E_n) = \frac{1}{\Gamma(m) (1-\rho_n^2)} \int_0^x \int_0^y \exp \left(-\frac{y+z}{1-\rho_n^2}\right)\nonumber\\
&\times \left(\frac{y z}{\rho_n^2}\right)^{\frac{m-1}{2}} I_{m-1}\left(\frac{2 \sqrt{y z \rho_n^2}}{1-\rho_n^2}\right) \prod_{\substack{k=2\\k\neq n}}^N\phi_k(z,y) dz dy.\label{ccdf2}
\end{align}
Finally, the proposition is proven by taking the derivative of the CDFs with respect to $x$.

\subsection{Proof of Corollary \ref{cor2}}\label{cor2_prf}
To assist with the simplification of the outage probability, we first apply the transformation $y \to mt/\ee$ to \eqref{cdf3}. In other words, we obtain the CDF in terms of SNRs. Therefore, we can simplify the PDF $f_{\h_n | E_n}(y | E_n)$, $n\neq 1$ (Eq. \eqref{con_pdf2}), as follows
\begin{align}
\lim_{\bar{\eta} \to \infty} f_{t | E_n}(t | E_n) &\overset{(a)}{\approx} \frac{m^2\exp(-\frac{mt}{\ee(1-\rho_n^2)}) (\frac{m}{\ee}t)^{mN-m-1}}{\Gamma (m+1)^N \prod_{k=2}^N (1-\rho_k^2)^m}\nonumber\\
&\quad\times \int_0^{\frac{m}{\ee}t} z^{m-1} \exp(-Sz) dz\nonumber\\
&\overset{(b)}{=}\frac{m^2\exp(-\frac{mt}{\ee(1-\rho_n^2)}) (\frac{m}{\ee}t)^{mN-m-1}}{\Gamma (m+1)^N \prod_{k=2}^N (1-\rho_k^2)^m}\nonumber\\
&\qquad\times\frac{1}{S^m} \gamma\left(m,\frac{mS}{\ee}t\right),
\end{align}
where $(a)$ follows by using \eqref{series} and keeping the sum's first term as well as from the fact $I_m(x) \approx (x/2)^2/\Gamma(m+1)$ for $x \approx 0$; $(b)$ follows from $\int_0^b \exp(-t)t^{a-1} dt = \gamma(a,b)$ \cite{GRAD}. Finally, as $\gamma(a,b) \approx b^a/a$ for $x \approx 0$, we end up with
\begin{align}
\lim_{\bar{\eta} \to \infty} f_{t | E_n}(t | E_n) &\approx \frac{m\exp(-\frac{mt}{\ee(1-\rho_n^2)}) (\frac{m}{\ee}t)^{mN-1}}{\Gamma (m+1)^N \prod_{k=2}^N (1-\rho_k^2)^m}.
\end{align}
The approximation for $n=1$ (Eq. \eqref{con_pdf1}) can be derived in a similar manner. Then, by approximating the Marcum-$Q$ function in \eqref{cdf3} as before and substituting the PDFs, the final expression follows after several algebraic operations.

\subsection{Proof of Proposition \ref{prop-div}}\label{proof_div}
At the SISO demodulator, the interleaving, modulation, and transmission over the FA block-fading channel convert the codewords $c$ and $c'$ onto points $\mathcal{E}$ and $\mathcal{E}'$ in a Euclidean space. For a given channel, performance is governed by the squared Euclidean distance $|\mathcal{E} - \mathcal{E}'|^2$, that is expressed as a sum of $\omega$ squared Euclidean distances associated to the $\omega$ non-zero bits of $c - c'$ \cite{gresset08}. For each $\omega$ squared distance, an equivalent channel model corresponding to the transmission of a BPSK modulation over one among $N$ channel realizations. By assuming $s$ divides $\psi$, we can write
\begin{align}
	|\mathcal{E} - \mathcal{E}'|^2 = \sum_{k=1}^{ \psi/s } f_k^2 + \sum_{j= \psi/s + 1}^{ N - \psi + 1} \ell_j^2,
\end{align}
where $\psi$ denotes the number of symbols out of the $N$ that are space-time rotated, and hence $N - \psi$ are the remaining unrotated symbols. The parameters $\ell_j^2$ are linearly dependent on $h_j$, which means they follow a non-central chi-square distribution with two degrees of freedom, conditioned on $h_1$. The parameters $f_k^2$ are linearly dependent on $||\qH_k ||^2$, where $\qH_k$ is the $s \times s$ matrix defined as
\begin{align}
\qH_k = \qS_k \text{diag}(g_{k+s(k-1)}, \dots, g_{k+sk-1}),
\end{align}
with $\qS_k$ the $s \times s$ sub-part of $\qS$ that rotates $s$ symbols. Hence, the parameters $f_k^2$ follow a non-central chi-square distribution with $2s$ degrees of freedom, conditioned on $h_1$. Although correlation exists between FA ports, the diversity order at the output of the demodulator only depends on the degrees of freedom of the chi-square distributions. The average amount of diversity that can thus be recovered at the output of the SISO demodulator is $\bar{s} = N/( \frac{\psi}{s} + N -\psi)$. At the output of the demodulator, there are $\mathcal{B} = \lfloor \frac{N}{\bar{s}} \rfloor$ non-central chi-square distribution laws on average (representing the $\mathcal{B}$ fading blocks). Now if we group $\mathcal{L}/\mathcal{B}$ bits into one nonbinary symbol, the code $\mathcal{C}$ becomes a length-$\mathcal{B}$ code constructed from an alphabet of size $2^{\mathcal{L}/\mathcal{B}}$. The minimum Hamming distance of the nonbinary code thus represents the diversity order the underlying binary code can achieve over the $N$-port FA block-fading channel. On the other hand, the minimum distance $d^*$ of an error correcting code with dimension $K$ and length $L$ is upper bounded by the Singleton bound as $d^* \leq L - K + 1$. By applying this to the nonbinary code described above, we obtain the modified Singleton bound \cite{malka99,knopp00} on the diversity order as the first term in the {\em min} function of \eqref{diversity}. The term $\left\lfloor \bar{s} d_H \right\rfloor$ is due to the fact that the minimum Hamming distance of the code defines the minimum Hamming weight of any pair of codewords, thus these bits are a limiting diversity factor in that, under ideal interleaving, they can see at most $\left\lfloor \bar{s} d_H \right\rfloor$ fading realizations. To avoid this situation, a code with a large minimum distance is usually selected for transmission.

\subsection{Examples of Space-Time Rotations}\label{app}
Rotations for different values of the average space-time combining factor $\bar{s}$ are now presented. For the case where $\bar{s} = N = 2$, the real $2 \times 2$ cyclotomic rotation is used given by \cite{viterbo_rotations}
\begin{equation}
\qS_1 = \left[\begin{array}{cc}
\cos(\chi) & \sin(\chi)\\
\sin(\chi) & -\cos(\chi)
\end{array}\right],
\end{equation}
with $\chi = 4.15881461$. We now consider the case $N=4$ and $R_c=1/2$. We thus need $\bar{s}=2$ to achieve maximum diversity four, and matrix $\qS_2$ below is used
\begin{equation*}
\qS_2 = \left[\begin{array}{cc}
\qS_1 & 0 \\
0 & \qS_1 \\
\end{array}\right].
\end{equation*}
However, if $R_c=1/3$ for the same channel, $\bar{s}=4/3$ is needed, matrix $\qS_3$ is used
\begin{equation*}
\qS_3 = \left[\begin{array}{cc}
\qS_1 & 0 \\
0 & \qI_2 \\
\end{array}\right].
\end{equation*}
Now if $R_c = 3/4$ and $N = 4$, all the symbols need to be combined as $\bar{s}=4$. In this case, the $4 \times 4$ Kr\"{u}skemper rotation \cite{viterbo_rotations} with normalized minimum product distance of $0.438993$ can be used.

\subsection{Proof of Proposition \ref{prop-pep}}\label{proof-pep}
We first consider the case where no space-time rotation is used, meaning a different modulated symbol is transmitted per FA port. The conditional PEP under ML decoding and Gray mapping, given the fading coefficients, can be approximated as \cite{boutros04}
\begin{equation}\label{pep1}
	\PP\{0 \to c \,|\, \qG\} \approx \delta Q \left(\sqrt{\zeta R_c \ee \sum_{k=1}^N w_k(c) h_k}\right),
\end{equation}
where $\ee$ is the SNR per bit at the receiver, $\delta = (4/b) (\sqrt{2^b} - 1)/\sqrt{2^b}$ and $\zeta = 3b/(2^b - 1)$ for $2^b$-QAM. This approximation, accurate for high SNR, is suitable for the analysis of diversity and coding gain. Next, we assume now that groups of $s$ modulated symbols out of the $N$ are fed to an $N \times N$ space-time rotation $\qS$ ($s$ divides $N$). By assuming the same rotation is used across the $N$ symbols, the power channel gain seen by partial Hamming weight $w_j(c)$ is given by $||\qH_k \qH_k^{\dagger}|| = \sum_{k=1}^s \kappa_k h_{(j-1)s + k}$.  The conditional PEP can now be written as
\begin{align}\label{pep2}
&\PP\{0 \to c \,|\, \qS,\qG\}\nonumber\\
&\approx \delta' Q \left(\sqrt{\zeta' R_c \ee \sum_{j=1}^{ N/s } w_j(c) \sum_{k=1}^s \kappa_k h_{(j-1)s + k}}\right).
\end{align}
Although the constellation resulting from the rotation of $s$ $2^b$-QAM is not exactly an $2^{sb}$-QAM, we assume that $\delta' = [4/(sb)] (\sqrt{2^{sb}} - 1)/\sqrt{2^{sb}}$ and $\zeta' = 3sb/(2^{sb}-1)$, as the rotation does not affect the energy of the modulated symbols. Finally, if we assume that the first $\psi$ symbols out of the $N$ are space-time rotated, and the remaining $N - \psi$ are not, a combination of \eqref{pep1} and \eqref{pep2} gives the expression in \eqref{approx1}.

\subsection{Proof of Corollary \ref{cor3}}\label{coro-pep-proof}
By using $Q(x) \leq \frac{1}{2} \exp(-\frac{x^2}{2})$ and by conditioning on $h_1$, we can upper-bound \eqref{approx1} as
\begin{align}
&\PP\{0 \to c \,|\, \qS, h_1\}\nonumber\\
&\leq \frac{\delta}{2} \E_{h_2, \dots, h_N} \Bigg\{\!\exp\!\Bigg(\!\!\!-\!R_c \ee\Bigg(\zeta\sum_{j=1}^{\psi/s} w_j(c) \sum_{k=1}^s \kappa_k h_{(j-1)s + k}\nonumber\\
&\hspace{4.5cm}+ \zeta' \sum_{j=\psi + 1}^N w_j(c) h_j\Bigg)\Bigg)\!\Bigg\}\nonumber\\
&=\frac{\delta}{2} \E_{h_2, \dots, h_N} \Bigg\{\prod_{j=1}^{\psi/s}\prod_{k=1}^s\exp(-R_c \ee \zeta w_j(c) \kappa_k h_{(j-1)s + k})\nonumber\\
&\hspace{3cm}\times\prod_{j=\psi + 1}^N \exp(-R_c \ee\zeta' w_j(c) h_j)\Bigg\}\nonumber\\
&=\frac{\delta}{2} \exp(-R_c \ee \zeta w_1(c) \kappa_1 h_1)\nonumber\\
&\qquad\times \prod_{j=\psi + 1}^N \frac{\exp\left(-\frac{\rho_j^2 h_1 R_c \ee\zeta' w_j(c)}{1 + (1-\rho_j^2) R_c \ee\zeta' w_j(c)}\right)}{1 + (1-\rho_j^2) R_c \ee\zeta' w_j(c)}\nonumber\\
&\qquad\times\prod_{j=1}^{\psi/s} \prod_{k=1}^s \frac{\exp\left(-\frac{\rho_{(j-1)s + k}^2 h_1 R_c \ee \zeta w_j(c)}{1 + (1-\rho_{(j-1)s + k}^2) R_c \ee \zeta w_j(c)}\right)}{1 + (1-\rho_{(j-1)s + k}^2) R_c \ee \zeta w_j(c)\kappa_k},
\end{align}
which follows from the moment generating function of a non-central chi-square distributed random variable with two degrees of freedom and non-centrality parameter $\frac{2h_1\rho_{(j-1)s + k}^2}{1-\rho_{(j-1)s + k}^2}$ with $(j-1)s + k \neq 1$. Then, for high SNRs, we can write
\begin{align}
&\lim_{\ee\to\infty}\PP\{0 \to c \,|\, \qS, h_1\}\nonumber\\
&\leq \frac{\delta}{2} \exp(-R_c \ee \zeta w_1(c) \kappa_1 h_1) \prod_{j=\psi + 1}^N \frac{\exp\left(-\frac{\rho_j^2 h_1}{1-\rho_j^2}\right)}{(1-\rho_j^2) R_c \ee\zeta' w_j(c)}\nonumber\\
&\qquad\times\prod_{j=1}^{\psi/s} \prod_{k=1}^s \frac{\exp\left(-\frac{\rho_{(j-1)s + k}^2 h_1}{1-\rho_{(j-1)s + k}^2}\right)}{(1-\rho_{(j-1)s + k}^2) R_c \ee \zeta w_j(c) \kappa_k}\nonumber\\
&=\frac{\delta}{2} \exp(-R_c \ee \zeta w_1(c) \kappa_1 h_1)\nonumber\\
&\quad\times \frac{(R_c \ee\zeta')^{\psi-N}\exp\left(-\sum_{j=\psi + 1}^N\frac{\rho_j^2 h_1}{1-\rho_j^2}\right)}{\prod_{j=\psi+1}^N(1-\rho_j^2) w_j(c)}\nonumber\\
&\quad\times \frac{(R_c \ee \zeta)^{1-\psi}\exp\left(-\sum_{j=2}^\psi \frac{\rho_j^2 h_1}{1-\rho_j^2}\right)}{\prod_{j=1}^{\psi/s} \prod_{k=1}^s(1-\rho_{(j-1)s + k}^2) w_j(c) \kappa_k}.
\end{align}
Then, by averaging out $h_1$, we end up with \eqref{asymp-pep}.

\end{document}